%% file: sosp.tex
\documentclass{article}

\usepackage[T1]{fontenc}
\usepackage{amsmath,amssymb,amsfonts,amsthm}
\usepackage{enumitem}
\usepackage{fullpage}
\usepackage{todonotes}
\usepackage[linesnumbered,boxed,vlined,ruled]{algorithm2e}
\newtheorem{definition}{Definition}
\newtheorem{theorem}{Theorem}

\newtheorem{lemma}{Lemma}
\newtheorem{example}{Example}
\newtheorem{corollary}{Corollary}

\renewcommand{\epsilon}{\varepsilon}
\newcommand{\out}{\mathcal{S}}
\newcommand{\vett}[1]{\mathbf{#1} }
\renewcommand{\a}{\vett a}
\renewcommand{\c}{\vett c}
\renewcommand{\d}{\vett d}
\newcommand{\x}{\vett x}
\newcommand{\y}{\vett y}
\newcommand{\ai}{\a_{-i}}
\renewcommand{\b}{\vett b}
\renewcommand{\t}{\vett t}
\newcommand{\z}{\vett z}
\newcommand{\bi}{\b_{-i}}
\newcommand{\ti}{\t_{-i}}
\newcommand{\ci}{\c_{-i}}
\newcommand{\di}{\d_{-i}}
\newcommand{\xii}{\x_{-i}}
\newcommand{\yi}{\y_{-i}}

\newcommand{\M}{\mathcal{M}}
\newcommand{\T}{\mathcal{T}}
\newcommand{\vgraph}{$k$-step OSP-graph}

\newcommand{\shp}{{\textsc{SP}}}
\newcommand{\notE}{\overline{E}}

\begin{document}

\title{
An Algorithmic Theory of Simplicity in Mechanism Design
} 
\author{Diodato Ferraioli
\and
Carmine Ventre
}


\date{}

\maketitle

\begin{abstract}
A growing body of work in economics and computation focuses on the trade-off between implementability and simplicity in mechanism design. The goal is to develop a theory that not only allows to design an incentive structure easy to grasp for imperfectly rational agents, but also understand the ensuing limitations on the class of mechanisms that enforce it. In this context, the concept of OSP mechanisms has assumed a prominent role since they provably account for the absence of contingent reasoning skills, a specific cognitive limitation. For single-dimensional agents, it is known that OSP mechanisms need to use  certain greedy algorithms. 
In this work, we introduce a notion that interpolates between OSP and SOSP, a more stringent notion where agents 
only plan a subset of their own future moves.  
We provide an algorithmic characterization of this novel class of mechanisms for single-dimensional domains and binary allocation problems, that precisely measures the interplay between simplicity and implementability. We build on this to show how mechanisms based on reverse greedy algorithms (a.k.a., deferred acceptance auctions) are 
algorithmically more robust to imperfectly rationality than those adopting greedy algorithms.
\end{abstract}

\section{Introduction}
\input{intro2}

\input{related}

\section{Preliminaries}
For the design of a mechanism, we need to define a set $N$ of $n$ \emph{selfish agents} and a set of feasible \emph{outcomes} $\out$.
Each agent $i$ has a \emph{type} $t_i \in D_i$, where $D_i$ is the \emph{domain} of $i$.
The type $t_i$ is assumed to be \emph{private knowledge} of agent $i$.
We let $t_i(X) \in \mathbb{R}$ denote the \emph{cost}
of agent $i$ with type $t_i$ for the outcome $X \in \out$.
When costs are negative, it means that the agent has a {profit} from the solution, called \emph{valuation}.

We would like to run a \emph{mechanism} in order to select an outcome and assign opportune payments.
I.e., the mechanism implements a pair $(f,p)$,
where $f$ (termed \emph{social choice function} or, simply, algorithm) maps the actions taken by the agents to a feasible solution in $\out$, and $p$ maps the actions taken by the agents to \emph{payments}. Note that payments need not be positive.
Each selfish agent $i$ is equipped with a \emph{quasi-linear utility function} $u_i \colon D_i \times \out \rightarrow \mathbb{R}$: for $t_i \in D_i$ and for an outcome $X \in \out$ returned by a mechanism $\M$, $u_i(t_i, X)$ is the utility that agent $i$ has for the implementation of outcome $X$ when her type is $t_i$, i.e., $u_i(t_i, X) = p_i - t_i(X)$.
In this work we will focus on \emph{single-parameter} settings, that is, the case in which
the private information of each bidder $i$ is a single real number $t_i$ and $t_i(X)$ can be expressed as $t_i w_i(X)$ for some publicly known function $w_i$.

In order to implement $(f, p)$, we design a game $\Gamma$ for the agents to play.
Specifically, $\Gamma$ is an imperfect-information, extensive-form game with perfect recall, defined in the standard way: $\mathcal{H}$ is a finite collection of partially ordered \emph{histories} (i.e., sequences of moves). At every non-terminal history $h \in \mathcal{H}$, one agent $i \in N$ is called to play and has a finite set of \emph{actions} $A(h)$ from which to choose. At some history $h$ it may be also possible that the mechanism $M$ plays by taking a random choice: specifically, we denote with $\omega(h)$ the realization of the mechanism's random choice at history $h$, and with $\omega = (\omega(h))_{h \in \mathcal{H} \colon \text{$M$ plays at $h$}}$ the mechanism's random choices along the entire game. Each terminal history is associated with an outcome $X \in \out$, and agents receive utility $u_i(t_i,X)$. We use the notation $h' \subseteq h$ to denote that $h'$ is a subhistory of $h$ (equivalently, $h$ is a continuation history of $h'$), and write $h' \subset h$ when $h' \subseteq h$ but $h\neq h'$. 
When useful, we sometimes write $h=(h', a)$ to denote the history $h$ that is reached by starting at history $h'$ and following the action $a \in A(h)$.

An \emph{information set} $I$ of agent $i$ is a set of histories such that for any $h, h' \in I$ and any subhistories  $\tilde{h} \subseteq h$ and $\tilde{h}' \subseteq h'$ at which $i$ moves, at least one of the following conditions holds: (i) there is a history $\tilde{h}^* \subseteq \tilde{h}$ such that $\tilde{h}^*$  and $\tilde{h}'$ are in the same information set, $A(\tilde{h}^*) = A(\tilde{h}')$, and $i$ makes the same move at $\tilde{h}^*$ and $\tilde{h}'$; (ii) there is a history $\tilde{h}^* \subseteq \tilde{h'}$ such that $\tilde{h}^*$  and $\tilde{h}$ are in the same information set, $A(\tilde{h}^*) = A(\tilde{h})$, and $i$ makes the same move at $\tilde{h}^*$ and $\tilde{h}$. We denote by $I(h)$ the information set containing history $h$. Roughly speaking, an information set collects all the histories that an agent is unable to distinguish.

A \emph{strategy} for a player $i$ in game $\Gamma$ is a function $S_i$ that specifies an action at each one of her information sets. When we want to refer to the strategies of different types $t_i$ of agent $i$, we write $S_i(t_i)$ for the strategy followed by agent $i$ of type $t_i$; in particular, $S_i(t_i)(I)$ denotes the action chosen by agent $i$ with type $t_i$ at information set $I$. We use ${\bf S}({\bf t}) = (S_i(t_i))_{i \in N}$ to denote the strategy profile for all of the agents when the type profile is ${\bf t}=(t_i)_{i\in N}$. An \emph{extensive-form mechanism} is an extensive-form game together with a profile of strategies ${\bf S}$.

Given an agent $i$, an information set $I$ is an \emph{earliest point of departure} between strategy $S_i$ and $S'_i$ in game $\Gamma$ if $I$ is on the path of play under both $S_i$ and $S'_i$ and both strategies choose the same action at all earlier information sets, but choose a different action at $I$. I.e., $I$ is the earliest information set at which these two strategies diverge. Note that for two strategies, there will in general be multiple earliest points of departure.

For an agent $i$ with preference $t_i$, strategy $S_i$ \emph{$k$-step obviously dominates} strategy $S'_i$ in game $\Gamma$ if, starting at any earliest point of departure $I$ between $S_i$ and $S'_i$, the outcome that maximize the utility of $i$ among the ones reachable following $S'_i$ at $I$ is weakly worse than the outcome that achieves the minimum utility among the ones reachable by following $S_i$ at $I$, where worst (best, resp.) cases are determined by considering any future play by other agents (including random choices of the mechanism) and any future play of agent $i$ that coincides with strategy $S_i$ ($S'_i$, resp.) in all information sets $I'$ following $I$ such that $i$ plays at most $k$ times between $I$ (excluded) and $I'$ (included).
In particular we simply say that $S_i$ \emph{obviously dominates} $S'_i$ when $k=\infty$, and that it \emph{strongly obviously dominates} $S'_i$ when $k = 0$.
If a strategy $S_i$ ($k$-step / strongly) obviously dominates all other $S'_i$, then we say that $S_i$ is \emph{($k$-step / strongly) obviously dominant}. If a mechanism implements $f$ by guaranteeing that each player has a ($k$-step / strongly) obviously dominant strategy, we say that it is \emph{($k$-step / strongly) obviously strategy-proof} (($k$-step / S) OSP).

\subsection{Round-Table Mechanisms}
Mackenzie \cite{mackenzie} proved that for OSP mechanisms, it is without loss of generality to consider mechanisms $\M$ with a specific format, named \emph{round table mechanisms}, defined as follows.
$\M$ is defined by a triple $(f,p,\T)$ where, as above, the pair $(f,p)$ determines the outcome of the mechanism, and $\T$ is a tree, called \emph{implementation tree}, such that:
\begin{itemize}[leftmargin=0.45cm, noitemsep, topsep=0pt]
	\item Every leaf $\ell$ of the tree is labeled with a possible outcome of the mechanism $(X(\ell), p(\ell))$, where $X(\ell) \in \out$ and $p(\ell) \in \mathbb{R}^n$;
	\item Each internal node $u$ in the implementation tree $\T$ defines the following:
    \begin{itemize}[leftmargin=0.6cm, noitemsep,topsep=0pt]
        \item An agent $i=i(u)$ to whom the mechanism makes some query. Each possible answer to this query leads to a different child of $u$.
        \item A subdomain $D{(u)}=(D_i{(u)}, D_{-i}{(u)})$ containing all types profiles that are \emph{compatible} with $u$ (or available at $u$), i.e., with all the answers to the queries from the root down to node $u$.
        Specifically, the query at node $u$ defines a partition of the current domain of $i$, $D_i{(u)}$ into $k\geq 2$ subdomains, one for each of the $k$ children of node $u$. Thus, the domain of each of these children will have, as for the domain of $i$, the subdomain of $D_i{(u)}$ corresponding to a different answer of $i$ at $u$, and an unchanged domain for the other agents. We also say that action at $u$ \emph{signals} the associated subdomain.
\end{itemize}
\end{itemize}
Observe that, according to the definition above, for every type profile $\b = (b_i \in D_i)_{i \in N}$ there is only one leaf $\ell = \ell(\b)$ such that $\b$ belongs to $D^{(\ell)}$.
Similarly, to each leaf $\ell$ there is at least a profile $\b$ that belongs to $D^{(\ell)}$.
This allows as to simplify the notation: indeed, we can define $\M(\b) = (X(\ell), p(\ell))$ for $\ell = \ell(\b)$. Similarly, we can simply write $f(\b)=(f_1(\b),\ldots,f_n(\b)) = X(\ell)$ and $p(\b)=(p_1(\b),\ldots,p_n(\b)) \in \mathbb{R}^n = p(\ell)$, and $u_i(t_i, \M(b_i, \bi)) = p_i(b_i, \bi) - t_i(f(b_i, \bi))$. For the single-parameter setting, considered in this work, we can further simplify the notation, by setting $t_i(X) = t_i f_i(\b)$ when we want to express the cost of a single-parameter agent $i$ of type $t_i$ for the output of social choice function $f$ when the actions taken by the agent lead to the leaf $\ell$ associated with bid vector $\b$.

Two type profiles $\b$, $\b'$ are said to \emph{diverge} (or to be separated) at a node $u$ of $\T$ if this node has two children $v, v'$ such that $\b \in D(v)$, whereas $\b' \in D(v')$. For every such node $u$, 
we say that 
$i(u)$ is the \emph{divergent agent} at $u$. We sometimes abuse notation and we say that two types $t_i$ and $b_i$ of the agent $i = i(u)$ diverge (are separated) at $u$.

A round-table mechanism $\M$ is OSP if for every agent $i$ with real type $t_i$,
for every vertex $u$ such that $i = i(u)$, for every $\bi, \bi'$ (with $\bi'$ not necessarily different from $\bi$),
and for every $b_i \in D_i$, with $b_i \neq t_i$,
such that $(t_i, \bi)$ and $(b_i, \bi')$ are compatible with $u$, but diverge at $u$,
it holds that $u_i(t_i, \M(t_i, \bi)) \geq u_i(t_i,\M(b_i, \bi')).$
Roughly speaking, OSP requires that, at each time step
agent $i$ is asked to take a decision that depends on her type, the worst utility that
she can get if she henceforth behaves according to her true type
is at least the best utility she can get by behaving differently.

We will extend this definition as follows. Let us first start with some useful definition.
For each agent $i$, and for each vertex $u$ in the implementation tree $\mathcal{T}$ such that $i = i(u)$, and every $u'$ along the path from $u$ to a leaf, we say that $u'$ belongs to the \emph{$k$-step limit} $\mathcal{L}_k(u)$ of $u$ if $u'$ is the $k$-th node along this path ($u$ excluded) in which $i$ is queried or, if $i$ is queried in less than $k$ nodes along this path, then $u'$ is the leaf. Moreover, for each agent $i$, and for each vertex $u$ in the implementation tree $\mathcal{T}$ such that $i = i(u)$, the \emph{$k$-step neighborhood} $\mathcal{N}_k(u)$ of $u$ consists of each node of $\T$ that appears in the path between $u$ and some $u' \in \mathcal{L}_k(u)$. Then, for each agent $i$ and for each vertex $u$ in the implementation tree $\mathcal{T}$ such that $i = i(u)$, two (not necessarily different) profiles $\a$ and $\b$ are said to be \emph{$k$-step unseparated} at $u$ if $\a$ and $\b$ do not diverge either at $u$ or in every node in the $k$-step neighborhood of $u$. Finally, for each agent $i$ and profile $\a$, and for each vertex $u$ in the implementation tree $\mathcal{T}$ such that $i = i(u)$ and $\a$ is available at $u$, the \emph{$k$-step equivalence class} of $\a$ at $u$ is $\Gamma^k_{u}(\a) = \{\b \colon \b \text{ is $k$-step unseparated from $\a$ at $u$}\}$.

Then, a round-table mechanism $\M$ is \emph{$k$-step OSP} if for every agent $i$ with real type $t_i$,
for every vertex $u$ such that $i = i(u)$, for every $\bi, \bi', \bi''$ (with $\bi', \bi''$ not necessarily different from $\bi$ and different from each other),
and for every $b_i, b'_i \in D_i$
such that $(t_i, \bi), (b_i, \bi') \in \Gamma^k_u(t_i, \bi)$,
while $(t_i, \bi)$ and $(b'_i, \bi'')$ diverge at $u$,
it holds that $u_i(t_i, \M(b_i, \bi')) \geq u_i(t_i,\M(b'_i, \bi'')).$
Roughly speaking, $k$-step OSP requires that, at each time step
agent $i$ is asked to take a decision that depends on her type, the worst utility that
she can get if she now behaves according to her true type
is at least the best utility she can get by behaving differently,
\emph{even if she later (i.e., after the next $k$ queries received by the mechanism) changes her mind and plays as if her type was an untruthful but still available one}.
As above, it is immediate to check that $k$-step OSPness reduces to OSPness when $k = \infty$, and has been dubbed strong OSPness for $k = 0$ by Pycia and Troyan \cite{pycia2021theory}.

Mackenzie \cite{mackenzie} proved that if there is an OSP mechanism implementing a social choice function $f$, then this can be also implemented by an OSP round-table mechanism. We next show that this claim holds also for $k$-step OSP.
\begin{theorem}
\label{thm:round-table}
 There is a $k$-step OSP mechanism implementing a social choice $f$ if and only if $f$ can be also implemented by a $k$-step OSP round-table mechanism.
\end{theorem}
The proof is essentially the same provided by Mackenzie \cite{mackenzie} for OSP, and hence we omit it here a formal argument, and only provide the main ideas.
\begin{proof}[Proof Sketch]
 Consider a general mechanism $\M$. If $\M$ is randomized, then fix a realization $\omega$ of the random choices of the mechanism, and consider only the set of histories of $\M$ that are compatible with this realization (this is what is denoted as \emph{$\omega$-derandomization} operation \cite{mackenzie}). Clearly, if the mechanism $\M$ is $k$-step OSP with respect the full set of histories, it will remain $k$-step OSP also with respect to the reduced set of histories. In other words, we are claiming that the randomized mechanism is $k$-step OSP if and only if each possible random realization of the mechanism must be $k$-step OSP.

Let us thus assume that $\M$ is a deterministic $k$-step OSP mechanism. Suppose that at an history $h$ the chosen player has only one available action. Clearly, we can safely remove this history without affecting the $k$-step OSPness of both this agent and remaining agents.
Suppose that the mechanism $\M$ also allows the agent to be \emph{absent-minded}, that is, to have two histories, one following the other, that belong to the same information set, representing in this way an agent that forgets to have played before. Hence, every possible action of other players occurring between the two histories is not relevant for the considered agent, while the action taken by this player must be the same at both histories. Hence, we can merge these two histories, without affecting neither the $k$-step OSPness of the current agent, nor the $k$-step OSPness of remaining agents (this corresponds to the \emph{untangling} operation \cite{mackenzie}).

Hence, we can assume that the mechanism is deterministic, and no agent is absent-minded. In other words, they have perfect information about the mechanism, and their own previous actions. Still, it is possible that the agents do not have perfect information about the actions taken by other agents, and different histories corresponding to different actions taken by agent $j$ are in the same information set of agent $i$. Observe that, in this case, the worst outcome that $i$ may play by taking a given action is not better than the worst that she can receive by taking this action when she knows about the action of $j$. Similarly, incomplete information cannot improve the best outcome that $i$ can receive. Hence, if the mechanism is $k$-step OSP with incomplete information, it must be $k$-step OSP also with complete information. And thus we can assume that each information set corresponds of a single history (this is the \emph{illumination} operation \cite{mackenzie}).

To conclude, since the mechanism $\M$ is deterministic and the players have complete information, it is sufficient to simply \emph{prune} that actions that will not be played by any strategy \cite{liosp}, and \emph{rename} actions with type subdomains for having a $k$-step OSP round-table mechanism, as desired.
\end{proof}


Due to Theorem~\ref{thm:round-table}, we will henceforth focus only on round-table mechanisms, and we will refer to them simply as mechanisms.

\section{Characterizing the Implementation Tree}
In this section we prove our first main result and characterize $k$-step OSP mechanisms in terms of extensive-form games where each agent is called to act at most $k+2$ times, with the limitations discussed above on the last time she plays. We note that the results in Section \ref{sec:ordered} hold for each outcome space whilst in Section \ref{sec:cyclemon} we restrict to binary outcomes.

\subsection{Almost-Ordered Mechanisms}\label{sec:ordered}
A mechanism $\M$ implementing social function $f$ with implementation tree $\T$ is said to be \emph{almost-ordered}, if for every agent $i$, every node $u$ of $\T$ such that $i = i(u)$, and every pair of profiles $\a$ and $\b$ separated at $u$ such that $f_i(\a) > f_i(\b)$ we have that  $c_i < d_i$ where $c_i = \max \{x_i \mid \exists \xii \colon (x_i, \xii) \in \Gamma_u^k(\a)\}$ and $d_i = \min \{x_i \mid \exists \xii \colon (x_i, \xii) \in \Gamma_u^k(\b)\}$. Roughly speaking a mechanism is almost ordered if for every pair $L$ and $ R$ of the subsets in which $D_i(u)$ is partitioned at $u$, either the mechanism assigns the same outcome to all types in $L \cup R$ or these two sets are ordered (i.e., the types in one of these two sets are all smaller than the types in the other set).

We next prove that every $k$-step OSP mechanism must be almost ordered. Let us first state the following useful lemma.
\begin{lemma} \label{lem:def}
   An extensive-form  mechanism $\M=(f,p)$ with implementation tree $\T$ is $k$-step OSP if and only if for all $i$, all vertices $u$ of $\T$ such that $i=i(u)$, and every pair of profiles $\a$ and $\b$ separated at $u$, the following holds for all $(c_i, \ci) \in \Gamma^k_u(\a)$:
    \begin{equation}
     \label{eq:osp_constraint}
     p_i(b_i,\bi) - p_i(a_i, \ai) \leq c_i\left(f_i(b_i, \bi)-f_i(a_i, \ai)\right).
    \end{equation}
\end{lemma}
\begin{proof}
	By definition, a mechanism $\M$ with implementation tree $\T$ is a $k$-step OSP mechanism for the social choice function $f$
	if and only if there exists a payment function $p$ such that $\M=(f, p)$ and, for every agent $i$ with real type $t_i$, and for every vertex $u$ such that $i = i(u)$, it holds that
	\[
	  p_i(b_i, \bi') - t_i f_i(b_i, \bi') = u_i(t_i, \M(b_i, \bi')) \geq u_i(t_i,\M(b'_i, \bi'')) =   p_i(b'_i, \bi'') - t_i f_i(b'_i, \bi'')
	\]
	for every $\bi, \bi', \bi''$ and for every $b_i, b'_i \in D_i$, with $b'_i \neq t_i$,
	such that $(t_i, \bi), (b_i, \bi') \in \Gamma^k_u(t_i, \bi)$, while $(t_i, \bi)$ and $(b'_i, \bi'')$ are compatible with $u$, but diverge at $u$. Since the true type of $i$ can be any value in $D_i$, then the mechanism is $k$-step OSP if and only if this is true for any pair $a_i, b_i \in D_i(u)$. \hfill
\end{proof}
\medskip
Henceforth, we will refer to condition~\eqref{eq:osp_constraint} as \emph{$k$-step OSP constraint}.

\begin{theorem}
\label{thm:ordered}
 A $k$-step OSP mechanism $\M$ implementing social function $f$ with implementation tree $\T$ is  almost-ordered.
\end{theorem}
\begin{proof}
 Suppose that $\M$ is not almost-ordered, and let $i$ and $u$ be respectively the agent and the node of $\T$ such that $i = i(u)$ and there are $\a$ and $\b$ with $f_i(\a) > f_i(\b)$ separated at $u$ such that $c_i > d_i$, where $c_i$ and $d_i$ are as defined above.
 Since $\a$ and $\b$ are separated at $u$, and there are, by definition of $c_i$ and $d_i$, $\ci$ and $\di$ such that $(c_i, \ci) \in \Gamma_u^k(\a)$ and $(d_i, \di) \in \Gamma_u^k(\b)$, then both the following $k$-step OSP constraints must be satisfied:
 \begin{align*}
  p_i(b_i,\bi) - p_i(a_i, \ai) & \leq c_i\left(f_i(b_i, \bi)-f_i(a_i, \ai)\right)\\
  p_i(a_i,\ai) - p_i(b_i, \bi) & \leq d_i\left(f_i(a_i, \ai)-f_i(b_i, \bi)\right).
 \end{align*}
 This leads to $c_i \leq d_i$, that contradicts our hypothesis. 
\end{proof}

\subsection{$k$-Limited Mechanisms for Binary Outcomes}\label{sec:cyclemon}
We henceforth focus on settings with binary outcomes, i.e., $\out = \{0, 1\}^n$.

Let us consider a mechanism $\M$ implementing $(f,p)$ with implementation tree $\T$. We next provide some useful definitions. We begin with certain subsets of type domains.
\begin{definition}[Suffix/Prefix of Type Domains]
Given a node $u \in \T$ and agent $i = i(u)$ we say that the domain $D_i(u)$ is \emph{prefix} if $\max \{t \in D_i(u)\} < \min \{t \in D_i \setminus D_i(u)\}$, i.e., it only contains the smaller types in the domain of $u$ (with larger types already removed in the queries preceding $u$). Similarly, we say that $D_i(u)$ is \emph{suffix} if $\min \{t \in D_i(u)\} > \max \{t \in D_i \setminus D_i(u)\}$.
\end{definition}
We now move to queries of particular interest.

\begin{definition}[Revelation, Extremal and (In)Effective Queries]
Given a node $u \in \T$ and a type $t \in D_i(u)$, $i=i(u)$, we say that the query at this node is a(n)
\begin{description}[nosep]
	\item[Revelation query] if agent $i$ is asked to reveal her type.
	\item[Extremal query] if agent $i$ is asked to separate one extreme (i.e., the minimum or maximum type) from the rest of her current domain $D_i(u)$.
	\item[Ineffective query] if for every $\xii$ available at $u$, and every $t', t'' \in D_i(u)$, $f_i(t', \xii) = f_i(t'', \xii)$ and $p_i(t', \xii) = p_i(t'', \xii)$. In words, regardless of how types are partitioned by this query, $i$ receives the same outcome and payment for every given profile of other agents' actions.
	\item[Strongly Ineffective query]  if for every $\xii, \xii'$ available at $u$ and every $t', t'' \in D_i(u)$, we have that $f_i(t', \xii) = f_i(t'', \xii')$ and $p_i(t', \xii) = p_i(t'', \xii')$. In words, a strongly ineffective query is one for which agent $i$ receives the same outcome and payment for each possible profile available at that history.
	\item[Only-$t$ effective query] if for each $\xii$ available at $u$ and every $t', t'' \in D_i(u)$ with $t', t'' \neq t$, $f_i(t', \xii) = f_i(t'', \xii)$ and $p_i(t', \xii) = p_i(t'', \xii)$, and there is $\yi$ available at $u$ such that $f_i(t, \yi) \neq f_i(t', \yi)$. In words, for each profile of other agents' actions, the outcome and payment received by $i$ is the same except for the type $t$.
	\item[Strongly only-$t$ effective query] if for each $\xii, \xii'$ available at $u$ and every $t', t'' \in D_i(u)$ such that $t', t'' \neq t$, $f_i(t', \xii) = f_i(t'', \xii')$ and $p_i(t', \xii) = p_i(t'', \xii')$, and there is $\yi$ available at $u$ such that $f_i(t, \yi) \neq f_i(t', \yi)$. That is, $i$ receives exactly the same outcome and payment in each profile available at $u$, except for the ones in which she has type $t$.
\end{description}
In our context, we will concentrate on (strongly) only-maximum (minimum, respectively) effective queries, that is, the (strongly) only-$t$ queries for which $t = \max \{t \in D_i(u)\}$ ($t = \min \{t \in D_i(u)\}$, respectively).
\end{definition}

We will next show that if $\M$ is $k$-step OSP, then it is without loss of generality, to assume that $\T$ has the following special structure. 

\begin{definition}[$k$-limited mechanism]
	A mechanism $\M$ implemented by a tree $\T$ is a \emph{$k$-limited} mechanism for each agent $i$, and for every path $P$ from the root of $\T$ to a leaf, one of the following properties holds:
	\begin{itemize}[topsep=0pt,noitemsep]
		\item $i$ is queried at most $k+1$ times along $P$;
		\item $i$ is queried at most $k+2$ times along $P$ \emph{and} at the node $u$ corresponding to the $(k+2)$-th query we have that: (i) either $|D_i(u)| = 2$ or $D_i(u)$ is prefix; and (ii) the query at $u$ is either a strongly ineffective revelation query, or a strongly only-maximum effective revelation query or it is an only-maximum effective extremal query that separates the maximum type in $D_i(u)$ from the rest of the domain;
		\item $i$ is queried at most $k+2$ times along $P$ \emph{and} at the node $u$ corresponding to the $(k+2)$-th query we have that: (i) $D_i(u)$ is suffix; and (ii) the query at $u$ is either a strongly ineffective revelation query, a strongly only-minimum effective revelation query or it is an only-minimum effective extremal query that separates the minimum type in $D_i(u)$ from the rest of the domain.
	\end{itemize}
\end{definition}
We then have the following theorem.
\begin{theorem}
 \label{thm:klimited}
 If there is a $k$-step OSP mechanism $\M$ that implements $(f, p)$ with implementation tree $\T$, then there is a $k$-limited $k$-step OSP mechanism $\M'$ implementing $(f, p)$.
\end{theorem}
The rest of this section proves Theorem~\ref{thm:klimited}. To this aim,  we first prove some preliminary properties of a $k$-step OSP mechanism $\M$, which amount to a version of the taxation principle for $k$-step OSP.

\begin{lemma}[Taxation Principle for $k$-step OSP] \label{lem:ineffective}
   Let $\M$ be a $k$-step OSP mechanism that implements $(f,p)$ with implementation tree $\T$. 
   For all $i$ and for each vertex $u$ in $\T$ such that $i = i(u)$, take 
   any three profiles $\a = (a_i, \a_{-i})$, $\c = (c_i, \ai)$, and $\d = (d_i, \a_{-i})$ such that (i) $a_i < c_i < d_i$, (ii) $\a, \c, \d \in \Gamma^k_u(\a)$, and (iii) there is $u' \notin {\cal N}_k(u)$ such that $i = i(u')$, $\a, \c$, and $\d$ are available at $u'$, and two among $\a$, $\c$, and $\d$ are separated at $u'$. We have that:
    \begin{enumerate}[nosep]
     \item if there are $\b$ and $\b'$ separated from $\a, \c, \d$ by $i$ along the path from $u$ to $u'$ such that $b_i > d_i$ and $b'_i < a_i$, then $f_i(\a) = f_i(\c) = f_i(\d)$ and $p_i(\a) = p_i(\c) = p_i(\d)$ (\textit{Outer-sandwich separations});
     \item if there is $\b$ separated from $\a, \c, \d$ by $i$ along the path from $u$ to $u'$ such that $a_i < b_i < d_i$, then $f_i(\a) = f_i(\c) = f_i(\d)$ and $p_i(\a) = p_i(\c) = p_i(\d)$ (\textit{Inner-sandwich separation});
     \item if there is $\b$ separated from $\a, \c, \d$ by $i$ along the path from $u$ to $u'$ such that $b_i > d_i$ then $f_i(\a) = f_i(\c)$ and $p_i(\a) = p_i(\c)$ (\textit{Top separation});
     \item if there is $\b$ separated from $\a, \c, \d$ by $i$ along the path from $u$ to $u'$ such that $b_i < a_i$ then $f_i(\c) = f_i(\d)$ and $p_i(\c) = p_i(\d)$ (\textit{Bottom separation}).
    \end{enumerate}
\end{lemma}
\begin{proof}
	We start by rewriting the $k$-step OSP constraint \eqref{eq:osp_constraint} for mechanism $\M$ 
	and agent $i$ with real type $t_i$. For every profile $(t_i, \ti)$, for every vertex $u$ such that $i = i(u)$ and $(t_i, \ti)$ is available at $u$, and every $\t' \in \Gamma^k_u(t_i, \ti)$, it holds that
	\begin{equation}
	 \label{cond:a_ti}
	p_i(\b) - p_i(\t') \leq t_i(f_i(\b) - f_i(\t'))
	\end{equation}
	for every $\b$ available at $u$ that diverges from $(t_i, \ti)$ at $u$. Hence, by taking $\ti = \ai$, if $t_i = a_i$, then we have that
	\begin{equation}
	\label{eq:case1}
	\begin{cases}
	 p_i(\b) - p_i(\a) \leq a_i(f_i(\b) - f_i(\a))& \text{if } \t' = \a;\\
	 p_i(\b) - p_i(\c) \leq a_i(f_i(\b) - f_i(\c))& \text{if } \t' = \c;\\
	 p_i(\b) - p_i(\d) \leq a_i(f_i(\b) - f_i(\d))& \text{if } \t' = \d;
	\end{cases}
	\end{equation}
	while if $t_i = c_i$, then we have that
	\begin{equation}
	\label{eq:case2}
	\begin{cases}
	 p_i(\b) - p_i(\a) \leq c_i(f_i(\b) - f_i(\a))& \text{if } \t' = \a;\\
	 p_i(\b) - p_i(\c) \leq c_i(f_i(\b) - f_i(\c))& \text{if } \t' = \c;\\
	 p_i(\b) - p_i(\d) \leq c_i(f_i(\b) - f_i(\d))& \text{if } \t' = \d;
	\end{cases}
	\end{equation}
	and if $t_i = d_i$, then we have that
	\begin{equation}
	\label{eq:case2bis}
	\begin{cases}
	 p_i(\b) - p_i(\a) \leq d_i(f_i(\b) - f_i(\a))& \text{if } \t' = \a;\\
	 p_i(\b) - p_i(\c) \leq d_i(f_i(\b) - f_i(\c))& \text{if } \t' = \c;\\
	 p_i(\b) - p_i(\d) \leq d_i(f_i(\b) - f_i(\d))& \text{if } \t' = \d.
	\end{cases}
	\end{equation}
	By taking instead $t_i = b_i$, $\ti = \bi$, $\t' = \b$, we have
	\begin{equation}
	\label{eq:case3}
	\begin{cases}
	 p_i(\a) - p_i(\b) \leq b_i(f_i(\a) - f_i(\b));\\
	 p_i(\c) - p_i(\b) \leq b_i(f_i(\c) - f_i(\b));\\
	 p_i(\d) - p_i(\b) \leq b_i(f_i(\d) - f_i(\b)).
	\end{cases}
	\end{equation}
	We distinguish three cases based on the relation between $b_i$, $a_i$, and $d_i$.
	\begin{description}
	 \item[Case 1 ($b_i > d_i$).]
	 Then from \eqref{eq:case1}--\eqref{eq:case3} we achieve that $f_i(\b) \leq \min\{f_i(\a), f_i(\c), f_i(\d)\}$ and \eqref{eq:case2bis} and \eqref{eq:case3} yield
	\begin{align}
	 L_1 = d_i(f_i(\a) - f_i(\b)) & \leq p_i(\a) - p_i(\b) \leq b_i(f_i(\a) - f_i(\b)) = U_1;
	\label{eq:cond1}\\
	 d_i (f_i(\c) - f_i(\b)) & \leq p_i(\c) - p_i(\b) \leq b_i(f_i(\c) - f_i(\b)).
	\label{eq:cond2}
	\end{align}

	Clearly, if $\a$ and $\c$ are never separated by $i$, then $f_i(\a) = f_i(\c)$.
	Suppose instead that there is a node $u'' \notin {\cal N}_k(u)$ (not necessarily the same as $u'$ defined in the statement) such that $\a$ and $\c$ are separated at $u''$ and $f_i(\a) \neq f_i(\c)$. Since $\a$ and $\c$ only differ in the type of agent $i$, then $i(u'') = i$. Then, by applying \eqref{cond:a_ti} with $\ti = \ai$, $t_i = a_i$, $\t' = \a$ and $\b = \c$, we have
	$$
	 p_i(\c) - p_i(\a) \leq a_i(f_i(\c) - f_i(\a)).
	$$
	Similarly, by applying \eqref{cond:a_ti} with $\ti = \ai$, $t_i = c_i$, $\t' = \c$ and $\b = \a$, we have
	$$
	 p_i(\a) - p_i(\c) \leq c_i(f_i(\a) - f_i(\c)).
	$$
	Hence, since $a_i< c_i$, we achieve that $f_i(\a) > f_i(\c)$ and thus
	\begin{equation}
	\label{eq:cond4}
	 a_i(f_i(\a) - f_i(\c)) \leq p_i(\a) - p_i(\c) \leq c_i(f_i(\a) - f_i(\c)).
	\end{equation}
	In order to satisfy \eqref{eq:cond2}, we need that $p_i(\c) = p_i(\b) + x_i(f_i(\c)-f_i(\b))$ for some $x_i \in [d_i, b_i]$. By replacing this value in \eqref{eq:cond4}, we then have that
	\begin{equation}
	\label{eq:cond4_new}
	 L_2 = a_i(f_i(\a)-f_i(\c)) + x_i(f_i(\c)-f_i(\b)) \leq p_i(\a) - p_i(\b) \leq c_i(f_i(\a)-f_i(\c)) + x_i(f_i(\c)-f_i(\b)) = U_2.
	\end{equation}
	Hence, in order to have that both \eqref{eq:cond1} and \eqref{eq:cond4_new} are satisfied, we need that 
	$U_2 \geq L_1$. That is, we need that
	$$d_i(f_i(\a)-f_i(\b)) = d_i(f_i(\a)-f_i(\c)) + d_i(f_i(\c)-f_i(\b)) \leq c_i(f_i(\a)-f_i(\c)) + x_i(f_i(\c)-f_i(\b)).$$
	Rearranging, we get 
	\begin{equation}
	 (d_i - c_i)(f_i(\a)-f_i(\c)) \leq (x_i - d_i)(f_i(\c)-f_i(\b)). \label{eq:option1}
	 \end{equation}
	Since $d_i > c_i$ and $f_i(\a)>f_i(\c)$, then the l.h.s. of \eqref{eq:option1} is larger than $0$.
	Moreover, since outcomes are binary, it must be the case that $f_i(\a) = 1$ and $f_i(\c) = 0$, and thus, since $f_i(\b) \leq f_i(\c)$, it must be that $f_i(\b) = f_i(\c) = 0$ meaning that the r.h.s. of  \eqref{eq:option1}  is $0$. Hence,
	\eqref{eq:option1} cannot be satisfied, and thus it cannot be that $f_i(\a) \neq f_i(\c)$.

	\item[Case 2 ($b_i < a_i$).] This case is symmetric with respect to the previous one, and it can be similarly proved by simply inverting the roles of $\a$ and $\d$.

	\item[Case 3 ($a_i < b_i < d_i$).] Then we have that the query at $u$ is not ordered. Since $\M$ is $k$-step OSP, then by Theorem~\ref{thm:ordered}, it must be almost-ordered, and thus it must be the case that $f_i(\b) = f_i(\a) = f_i(\c) = f_i(\d)$, as desired.
	\end{description}
	The arguments above imply that the outcomes are the same for every profile in $\Gamma_u^k(\a)$. The fact that also payments are the same immediately follows, since either the two profiles are never separated by $i$ (and hence they receive the same outcome and the same payment) or they are separated by $i$, and hence they need to receive the same payment otherwise the mechanism fails to be even simply strategyproof.
	
	To conclude the proof, we simply observe that when outer-sandwich separations occur then Case 1 and 2 hold simultaneously; inner-sandwich separation corresponds to Case 3 whilst top (bottom, respectively) separation corresponds to Case 1 (2, respectively). 
\end{proof}

\input{caratt.tex}

\section{Approximation Guarantee of $k$-step OSP Mechanisms}
\input{apps.tex}

\section{Conclusions}
In this work, we have studied the algorithmic robustness of OSP to agents that are not able to perform contingent reasoning and think about their future actions. Specifically, we introduce a novel notion, termed $k$-step OSP, that smoothens the notions of OSP (where  absence of contingent reasoning is the only cognitive limitation) and SOSP (where in addition agents are unable to think about any of their future actions) by maintaining the assumption that agents are not able to think contingently but allowing them a foresight of $k$ self moves ahead. We provide an algorithmic characterization of these mechanisms for single-dimensional agents and binary outcomes, via the introduction of a new cycle monotonicity toolkit. We apply our characterization to downward-close maximization problems and prove that the performance of OSP can deteriorate when $k$ is small in comparison to the type space of the agents. En route, we prove that reverse greedy algorithms are more robust than greedy algorithms to this worsening of the performances.

A natural open problems left by this work is to understand the extent to which the findings above hold for non-binary allocation problems, and quantify the limitations of limited planning horizons for other optimization problems, such as scheduling related machines that is now fully understood for OSP \cite{ec2023}. We highlight that our framework can also be adopted for non-binary allocation problems to prove that the query from the $(k+2)$-th onward must be limited in some way (specifically, \eqref{eq:option1} must hold), but this does not necessarily lead to mechanisms with only $k+2$ queries (as highlighted by the example in Appendix~\ref{apx:num_query}) since the Taxation Principle for $k$-step OSP would be less clean as the relative weight of types and outcomes would come to the fore. 


\bibliographystyle{plain}
\bibliography{ospb}
\newpage
\appendix

\section{On the Definition of $k$-step Simple Mechanisms}
\label{apx:pycia}
Pycia and Troyan \cite{pycia2021theory} defined the concept of \emph{$k$-step simple mechanism}, that resembles our definition of $k$-step OSP mechanisms. We here further discuss about similarities and differences among these concepts.

Let us start by describing the concept of $k$-step simple mechanism.
To this aim, we need to define the concept of \emph{strategic collections} that consist in a planning about the action to play at the current node, along to the ones to play at a certain number of nodes following it. In playing a mechanism, the agent defines a strategic collection for each node in which she interacts with the mechanism. However, the agent is allowed to choose strategic collections that are not coherent with each other: e.g., the strategic collections selected for node $u$ may state that the agent play a certain action at the current node and plans that at the following node $u'$ she will play action $a$; however, the strategic collection at node $u'$ may state that instead at the current node a different action $a'$ is instead played. In other words, in \cite{pycia2021theory} the plan done at previous node does not constrain the action taken at the following nodes, while in our definition we will only consider ``coherent'' strategic classification, in which what has been planned is actually played. (See also Example \ref{example:auctions}.)

Note that the definition of Pycia and Troyan \cite{pycia2021theory} leads to a definition of OSP ($k$-step simple mechanisms for $k=\infty$) that is slightly different from the original one by Li \cite{liosp}, since it is based on strategic collections in place of strategies: indeed, even if for an $\infty$-step simple mechanism the dominant strategic collection is coherent, and hence corresponds to a strategy, this strategy is required to be better even with respect to non-coherent strategic collections, that do not correspond to strategies (while in Li \cite{liosp}, the dominant strategy is compared only against strategies).

Moreover, the definition of $k$-step simple mechanisms also fails to capture a gradual transition from $OSP$ to $SOSP$. This is already observed by Pycia and Troyan \cite{pycia2021theory} about binary allocation problems, for which they showed that $1$-step simplicity is equivalent to OSP. Actually, we observe that this actually occurs for each problem. Indeed, since each future action is not constrained, then it must be the case that the each deviation possible in an OSP mechanism is also possible in a $k$-simple mechanism for each $k > 0$: in other words, if we try to compute an equivalent of the OSP graph \cite{esa19} for this mechanism, we end up with exactly the same graph.

\input{num_query}

\end{document}

%% file: intro2.tex

A set of agents need to determine an outcome that will affect each of them. In these cases, a mechanism is designed to interact with the agents and compute an outcome based on the decisions made during the interaction. In its more general form, the mechanism can be modelled as an extensive-form game. To guarantee properties on the quality of the solution computed (e.g., optimality for an objective function of interest) mechanisms are required to provide incentives for the agents to behave in a predictable and desirable way. For example, in dominant-strategy incentive-compatible mechanisms, it is pointwise optimal for the agents to behave correctly and truthfully report their private information to the mechanism. This property can however be too weak for certain agents, as observed experimentally \cite{ausubel2004,kagel87}. In particular, different extensive-form implementations of the mechanism can lead to different degrees of strategic confusion. Consider, for example, software agents that take decisions during the execution of the mechanism. These agents could take actions that are irrational from the economic point of view when they have been ``badly'' programmed, either because the programmer misunderstood the incentive structure in place or due to computational barriers preventing from comparing payoffs for each strategy adopted by the other agents \cite{book}. A less complex decision process would be guaranteed by associating one outcome/payoff to each possible action taken by the agent, since she would simply need to rank her own actions independently of the behavior of the other agents. This is the idea behind obviously strategyproof (OSP) mechanisms \cite{liosp}. OSP is a stronger notion of incentive-compatibility that takes a conservative view and requires honesty to be an obviously dominant strategy: the worst payoff that can be achieved by choosing it is not worse than the best possible payoff that can be obtained with a different strategy, where worst and best are chosen over the possible strategies of the remaining agents.


OSP is shown to capture the incentives of a specific form of imperfect rationality --- absence of contingent reasoning skills. From the computational point of view, OSP is known to be intimately linked with greedy algorithms for single-dimensional agents: OSP is equivalent to certain well-defined combinations of greedy algorithms that are suitably monotone in the agent's private information \cite{wine21,ec2023}. This algorithmic lens allows to focus on the quality of the solutions output by OSP mechanisms, measured in terms of the approximation guarantee to a given objective function, and conclude that these monotone greedy algorithms can perform well for some binary allocation problems \cite{wine21} but, in general, less well for richer solution spaces \cite{ec2023}. 


Does OSP encapsulate strategic simplicity in mechanism design? Consider agents with 
limited foresight, who can only devise a plan for their next $k$ moves (their \emph{planning horizon}) but cannot anticipate what they will do after that. It may not be simple enough for agents to go beyond that horizon in their reasoning; in chess, for example, if white can always win, any winning strategy is obviously dominant but requires to look at many moves in the future meaning that the strategic choices in chess are far from obvious \cite{pycia2021theory}. Similarly, software agents may need more data to be ``retrained'' and get a more granular picture of the scenarios that go beyond their $k$-th future move and plan accordingly. What is simplicity in mechanism design for these agents? OSP would result simple for agents with infinite planning horizons (e.g., agents who need not retrain and have all the data available from the beginning). A solution concept called Strong OSP (SOSP) has been recently introduced in the literature to define simplicity for agents who cannot forecast any of their own future moves \cite{pycia2021theory} (e.g., those who need to retrain after each move). Is there a notion that interpolates between these two extremes? Do we, in case, get a large class of implementable mechanisms the longer the planning horizons of the agents are?

\subsection{Our Contribution}
We introduce the notion of $k$-step OSP mechanisms to capture the incentive compatibility of agents who decide their next move by only planning for the subsequent $k$ steps and leave the remaining moves undecided. Thus, OSP corresponds to the case in which $k=\infty$ since, when determining their next move, agents reason about all their subsequent actions in the extensive-form mechanism. SOSP, on the contrary, corresponds to the case in which $k=0$ in that no extra future action is planned for in addition to the one under consideration. We fully characterize $k$-step OSP mechanisms with transfers for single-parameter agents and binary allocation problems:
\smallskip\noindent
\begin{center}
\begin{minipage}{0.85\textwidth}
	\textbf{Main Theorem 1 (informal).} \emph{A mechanism is $k$-step OSP for single-parameter agents and binary allocation problems if and only if each agent receives at most $k+2$ ``OSP queries'', the $(k+2)$-th (if any) being ``payoff-neutral up to one type''.}
\end{minipage}
\end{center}
\smallskip
In the informal statement above, by ``OSP query'' we mean that the extensive-form mechanism must satisfy the ensuing OSP constraints, whereas by ``payoff-neutral up to one type query'' we mean that the payoff for the queried agent $i$ is essentially determined for all her moves but potentially one; this move corresponds to a single type of $i$ for which her outcome can still be undecided at that point of the game. Whilst it is expected that the queries must be OSP, it is surprising that there is clean connection between the number of possible queries that each agent can be asked and her planning horizon, with up to $k+1$ ``unrestricted'' queries to each agent $i$ and one closing extra query that not only can affect the outcome that the other players will receive but also $i$'s outcome (albeit very limitedly)\footnote{Note indeed that this connection is neither an ``innate'' nor a straightforward property of $k$-step OSP mechanisms, since for larger outcome spaces there may be $k$-step OSP mechanisms with more than $k+2$ queries (cf. Appendix~\ref{apx:num_query}).}. This result thus shows that our notion fully captures the trade-off between simplicity and implementability. The more sophisticated the decision making of the agents is, the more mechanisms are incentive compatible. In particular, our notion gives rise to a fully nested class of mechanisms since every $k$-step OSP mechanism is also $k'$-step OSP for each $k' > k$.

To prove this characterization, we develop the theoretical foundations to study $k$-step OSP mechanisms. We prove that it is without loss of generality to focus on a certain family of implementations, called round-table mechanisms, and we find a structural property of $k$-OSP constraints that limits the class of queries that the mechanism can make\footnote{These results do not require the outcome space to be binary.}. We subsequently give a version of the Taxation Principle for $k$-OSP mechanisms by proving that from a certain point of the extensive-form game, outcomes (and payments) need to be constant for the agent being queried. Few additional structural properties are then proved to bound the number of queries to $k+2$ and determine what is the power of the last query.

This result complements the results by Pycia and Troyan \cite{pycia2021theory} for binary allocation problems from two perspectives. Firstly, our characterization says that each agent has two payoff relevant queries in an SOSP mechanism. Under an assumption of rich domains, Pycia and Troyan \cite{pycia2021theory} prove instead that agents only take one payoff-relevant move in SOSP mechanisms. Thus, our result highlights that for general single-dimensional type domains, there can be a second query which affects the outcome of the queried agent in a very specific and limited way. Secondly, Pycia and Troyan \cite{pycia2021theory} consider a notion of $k$-step dominance (see Appendix~\ref{apx:pycia}), which is very related to ours, in that agents think about $k$ future self moves when taking a decision. It turns out that their definition is technically different, with consequences on the boundaries between simplicity and implementability. In the work of Pycia and Troyan \cite{pycia2021theory}, although agents take into account their subsequent $k$ moves they do not commit to them but rather follow the principle of ``cross[ing] that bridge when you come to it'' (i.e., make choices as they arise) \cite{savage54}. It is thus assumed that agents in their decision making process optimistically complete the path of play beyond their planning horizon. To some extent, our notion formalizes the ``look before you leap'' (i.e., create a complete contingent plan for the possible $k$ future decisions one may have) model of decision making  \cite{savage54} and take a pessimistic perspective by considering the worst possible way in which the agent can complete her path of play beyond her planning horizon. Consequently, our approach is more restrictive but more robust to decisions made outside the agents' planning horizons. (Please see Example \ref{example:auctions} for an illustration of the differences between the two notions.) Importantly, within their modelling of decision making, Pycia and Troyan \cite{pycia2021theory} prove that one future move is enough for $1$-step simple mechanisms, for which honesty is one-step dominant, to collapse to OSP mechanisms. In fact, agents can at each step reconsider whether the plan that was 1-step dominant at the previous decision point is still the right way to go; intuitively, this freedom allows the agents to reconstruct an obviously dominant strategy step by step. Our notion instead draws a more fine grained boundary between the mechanisms that can be implemented and the behavioral model of the agents interacting with them.

\begin{example}[Ascending price auctions]\label{example:auctions}
	Consider an ascending auction for a single good, where at each non-terminal information set, an agent is called to play and has two actions, Stay In or Drop Out. The payoff of an agent $i$ is equal to the agent’s valuation $v_i$ minus her payment if the agent is allocated the good and $0$ otherwise. The price for the good weakly increases along each path of play. The auction ends when there is only one agent who has not dropped out; she wins the good and pays the price associated to the last time she moved. We remark that the ascending auction is OSP because the strategy of staying in as long as the current price is below the agent’s valuation is obviously dominant, as shown in \cite{liosp}.
	
	A one-step dominant strategy at information set $I^*$ is the following. If $p(I^*)$, the price at $I^*$, is not higher than $v_i$, then Stay In and Drop Out at the subsequent information set $I$ along a path of play. If the price at $I^*$ is higher than $v_i$, then Drop Out at $I^*$ and the subsequent information set $I$. This is one-step dominant because in both cases, the minimum from following the strategy is $0$ (dropping out eventually) which is not worse than other strategies at $I^*$ which can only guarantee a non-positive payoff. Importantly, in her decision process, agent $i$ is ignoring future paths of play wherein the agent will not drop out. However,  player $i$ must not necessarily drop out at $I$ but can still stay in if the price at $I$ is not higher than $v_i$. In other words, agent $i$ can eventually have inconsistent plans (see also Remark 1 in \cite{pycia2021theory}). Moreover, there are paths where the agent could, beyond her planning horizon, stay in even for prices higher than $v_i$.
	
	A $1$-step obviously dominant strategy instead mimics the aforementioned obviously dominant strategy but limits it to a planning horizon of two moves. So the agent would look at both $p(I^*)$ and $p(I)$, prices at $I^*$ and $I$, and Stay In (Drop Out, respectively) at $I' \in \{I^*, I\}$ if $p(I') \leq v_i$ ($p(I') > v_i$, respectively). Incidentally, and differently from above, our notion leads to consistent strategies since the decision of agent $i$ at $I$ will follow the same recommendation dictated by the strategy at $I^*$. This strategy is $1$-step obviously dominant only when the minimum payoff that can be reached by completing it in any possible way would be better than a deviation. So, for example, if the type domain of agent $i$ were \$1, \$2, \$3, \$4 and \$5, and her type were \$4 then, by following the strategy, agent $i$ would have stayed in for prices \$2 and \$3 (for her first two moves). To evaluate the $1$-step obvious dominance of the strategy, agent $i$ would then take the worst possible path of play that her future self could take, e.g., the decision to stay in for a price of \$5 leading to a payoff of $-1$. The best possible payoff by deviating is $0$ (e.g., dropping out immediately). 
\end{example}

To further explore the boundary between implementabilty and simplicity, we study the extent to which the approximation guarantee of OSP deteriorates when future self moves are not fully accounted for in decision making. We build on our first theorem to characterize what social choice functions can be implemented by $k$-step OSP mechanisms. Towards this end, we define a version of cycle monotonicity for $k$-step OSP; as in the cases of SP and OSP (amongst others) cycle monotonicity is a useful tool in that it allows to separate the allocation function (whose approximation guarantee we want to bound) and the payment function. The construction of this useful toolkit requires some further work to conveniently rewrite $k$-step OSP constraints. We are then able to algorithmically characterize $k$-step OSP mechanisms, by relating the cycles of the OSP graph \cite{MOR22} with those of the $k$-step OSP graph we define.
\smallskip\noindent
\begin{center}
	\begin{minipage}{0.85\textwidth}
		\textbf{Main Theorem 2 (informal).} \emph{There is a $k$-step OSP mechanism implementing social choice function $f$ if and only if $f$ is two-way greedy with $k$-limitable priority functions.}
	\end{minipage}
\end{center}
\smallskip
As hinted above, the algorithmic nature of OSP for binary allocation problems can be explained in terms of two-way greedy algorithms \cite{wine21}. These are algorithms that build the eventual solution by either greedily including agents with low cost (high valuation, respectively) --- termed forward greedy algorithm --- or by greedily excluding agents that have a high cost (low valuation, respectively) --- termed reverse greedy algorithm --- for minimization (maximization, respectively) problems.\footnote{In general, different agents can face different directions; in specific circumstances, it is possible to move from one direction to the other even for single agents, cf. Section \ref{sec:algo}.} These algorithms adopt adaptive priority functions (that is, priority functions that depend on the past choices made by the algorithm) to define the greedy order in which agents have to be processed. The restriction from OSP to $k$-step OSP limits the class of priority functions that can adopted -- it must be the case that each agent $i$ is at the top of the priority list for no more than $k+2$ times excluding the instances in which $i$ is consecutively at the top. These are $k$-limitable priority functions. Intuitively, this means that the greedy algorithm must be ready to decide whether an agent is included or excluded from the solution in no more than $k+2$ steps, thus clearly demarcating the algorithmic limitations that $k$-step OSP imposes. 

We conclude this work by exploring the power of this class of algorithms. We focus here on $p$-systems,
a class of problems for which it is known that greedy algorithms can compute a solution with approximation at most $p$.
To what extent can this result be replicated for $k$-step OSP mechanisms? It turns out that the size of the type domain of the agents and the format of the greedy algorithm play a crucial role. By letting $d$ denote half the size of the agents' type domain, we obtain the following tight result (assuming here for simplicity of exposition that the type domain has even size).
\smallskip\noindent
\begin{center}
	\begin{minipage}{0.85\textwidth}
		\textbf{Main Theorem 3 (informal).} \emph{There is a $p$-approximate $k$-limitable reverse greedy algorithm for $p$-systems whenever $k \geq d-2$. There is an instance of a $p$-system, even for $p = 1$ (i.e., a weighted matroid) for which no $k$-limitable two-way (forward, respectively) greedy algorithm with bounded approximation for $k< d-2$ ($k<2d-3$, respectively).}
	\end{minipage}
\end{center}
\smallskip
We can then conclude that, in the worst case, there is no gradual degradation of the performances of $k$-step OSP mechanisms; as the value of $k$ decreases, there is a stark dichotomy. Essentially, the two-way greedy algorithm needs to traverse the domain of each agent to compute good solutions when the types in the domain are sufficiently far apart, meaning that $k$ must be large enough. If it is conceivable to imagine that were it possible to cluster the domain of the agents around $k$ centres, then better results would be possible.

Interestingly, reverse greedy needs half the number of priority functions of forward greedy. It was already known that forward greedy is less robust than reverse greedy to economic desiderata (such as, individual rationality) of OSP mechanisms, due to the characterization of OSP payments by Ferraioli and Ventre \cite{FerraioliV23}. Our result further shows that reverse greedy is not only economically but also algorithmically more robust than forward greedy to imperfect rationality.

%% file: related.tex
\section{Related Work}
Since its inception \cite{liosp}, OSP has attracted the interest of both computer scientists and economists. Many papers \cite{ashlagigonczarowski,troyan2019obviously,thomas2021classification,mandal2022obviously} study OSP for stable matchings and provide an impossibility result in general as well as a suitable mechanism under some assumptions. Bade and Gonczarowski \cite{badegonczarowski} further study this concept for a number of settings, including single-peaked domains. Single-peaked preferences have been object of further work on OSP mechanisms \cite{AMN19,AMN20}. Pycia and Troyan \cite{pycia2} study OSP mechanisms in domains where monetary transfers are not allowed and provide a useful characterization in this setting.

For the setting where money is permitted, characterizations of OSP mechanisms have been provided by Ferraioli et al. \cite{wine21} for binary allocation problems and by Ferraioli and Ventre \cite{ec2023} for general single-dimensional problems. These characterizations are based on the OSP cycle-monotonicity technique introduced by Ferraioli et al. \cite{ferraioli2018approximation} along with some bounds on the approximation for machine scheduling and set system problems (see also \cite{esa19,wine19}).

OSP has been studied also under different restrictions on the agents' behavior during the execution of the mechanism. E.g., \emph{monitoring} has been shown to help the design of OSP mechanisms with a good approximation ratio in various mechanism design domains \cite{ferraioli2017obvious}. Kyropoulou and Ventre \cite{kyropoulou2019obviously} builds forth on this by studying machine scheduling in the absence of monetary payments and showing the power of OSP in this setting. Under the assumption that non-truthful behaviour can be detected and penalised with a certain probability, it is proved that every social choice function can be implemented by an OSP mechanism with either very large fines for lying or a large number of ``verified'' agents.

OSP mechanisms for single-minded combinatorial auctions are studied by de Keijzer et al. \cite{bartmaria}. Mackenzie \cite{mackenzie} presents a revelation principle that states that every social choice function implementable through an OSP mechanism can be implemented using a certain structured OSP protocol where agents take turns making announcements about their valuations. We prove a similar result here. Finally, Ferraioli and Ventre \cite{FerraioliV23} provide explicit formulas for the payment functions of OSP mechanisms, giving a constructive proof to the existential existence of payments guaranteed by cycle monotonicity. Clearly, their results can be imported to our notion of $k$-step OSP and confirm, that from the economic point of view, reverse greedy is more robust than forward greedy.

Several variants of OSP have been defined.
$k$-OSP \cite{ferraioliventreboundedJ} is a notion that interpolates between OSP and the classical strategyproofness (SP) for perfect rationality. The smoothing parameter $k$ measures the number of other agents, each can think about contingently when taking decisions. Thus, OSP corresponds to $k=0$ whereas $k$ is equal to the number of agents in the mechanism for SP. This work proves that $k$-OSP mechanisms are not much better than OSP mechanisms in approximating the makespan unless $k$ is very large. In a sense, this suggests that the limits imposed by imperfect rationality do not quickly evaporate and, as long as approximation is concerned, one can focus on OSP without loss of generality.

Another related line of work concerns the notion of non-obviously manipulable (NOM) mechanisms \cite{Troyan2020} where the absence of contingent reasoning skills is advocated to limit the misbehavior of agents as opposed to limiting strategyproofness (as in OSP). Recent work has provided a characterization single-dimensional domains \cite{AdKV2023a,AdKV2023b} and a general recipe for their design \cite{AdKV2024}.

%% file: caratt.tex
\newcommand{\verk}{{\cal O}_{i, f, \T}^k}
\newcommand{\verinf}{{\cal O}_{i, f, \T}^{\infty}}

We say that a mechanism $\M$ with implementation tree $\T$ is \emph{almost $k$-limited} if $\T$ is such that for each agent $i$ and every path $P$ from the root of $\T$ to a leaf, the following conditions are satisfied:
\begin{itemize}[nosep]
 \item for $u \in P$ corresponding to the $(k+2)$-th query to $i$, if we denote with $a_i$ and $d_i$ the minimum and the maximum of the types of $i$ available at $u$ respectively, we must have that
 \begin{itemize}[nosep]
 \item if the domain $D_i(u)$ of $i$ at $u$ contains at least three types and it is neither prefix nor suffix (i.e., in the previous $k+1$ queries to $i$, either $i$ separated from $D_i(u)$ types $b_i$ and $b'_i$ such that $b_i > d_i$ and $b'_i < a_i$, or $i$ separated from $D_i(u)$ type $b_i$ such that $a_i < b_i < d_i$), then the $(k+2)$-th query is ineffective;
 \item if $D_i(u)$ contains at most two types or it is prefix (i.e., in the previous $k+1$ queries to $i$, $i$ separated from $D_i(u)$ only types $b_i$ such that $b_i > d_i$), then the $(k+2)$-th query is either ineffective or it is only-maximum effective;
 \item if $D_i(u)$ contains at least three types and it is suffix (i.e., in the previous $k+1$ queries to $i$, $i$ separated from the domain only types $b_i$ such that $b_i < a_i$), then the $(k+2)$-th query is either ineffective or it is only-minimum effective.
\end{itemize}
 \item every $u \in P$ corresponding to the $q$-th query to $i$ with $q \geq k+3$ are {ineffective}. 
\end{itemize}

Considering the first query along a path $P$, Lemma~\ref{lem:ineffective} restricts the $(k+2)$-th query on $P$ as in the first bullet point (the three cases corresponding to sandwich, top, and bottom separations, respectively). Moreover, given these properties, it is not hard to see that the subsequent queries must be ineffective (as requested by the second condition of almost $k$-limited mechanisms). Thus, Lemma~\ref{lem:ineffective}  can be restated as follows.
\begin{corollary}
 If a mechanism $\M$ that implements $(f,p)$ with implementation tree $\T$ is $k$-step OSP, then $\T$ is almost $k$-limited.
\end{corollary}

%

We now are going to prove further limits to the effectiveness of the $(k+2)$-th query and subsequent ones, by showing that if a query is ineffective for a pair of separated types then it is in fact strongly ineffective for those types.

\begin{lemma}[Strong Ineffectiveness upon Separation]
\label{lem:strongly}
 If a mechanism $\M$ that implements $(f,p)$ with implementation tree $\T$ is $k$-step OSP, then for every $i$, and every path $P$ of $\T$, every node $u \in P$ such that $i(u) = i$, if there are $t, t' \in D_i(u)$ such that for every $\xii$ available at $u$ we have that $f_i(t, \xii) = f_i(t', \xii)$, and $t, t'$ are separated at $u$, then for every $\xii, \xii'$ available at $u$ we have that $f_i(t, \xii) = f_i(t', \xii')$ and $p_i(t, \xii) = p_i(t', \xii')$.
\end{lemma}
\begin{proof}
  Suppose that the claim does not hold, and there are $\y = (t, \yi)$ and $\y' = (t', \yi')$ for some $\yi$ and $\yi'$ available at node $u$, such that $f_i(t, \yi) \neq f_i(t', \yi')$.

  Consider then the profiles $\z = (t, \yi')$ and $\z' = (t', \yi)$. Since $\yi$ and $\yi'$ are available at node $u$, and $t, t'$ are separated at $u$, then it must be the case that $\z$ and $\z'$ also are separated at $u$ by $i$. Moreover, by hypothesis, $f_i(\z) = f_i(\y')$ and $f_i(\z') = f_i(\y)$. Since the mechanism is $k$-step OSP, then it must be the case that all the following $k$-step OSP constraints\footnote{These are actually OSP constraints, meaning that the lemma holds true for the strongest $k$-step OSP notion.} are satisfied:
 \begin{align}
  p_i(\y') - p_i(\y) \leq t(f_i(\y') - f_i(\y)); \label{eq:const1}\\
  p_i(\y) - p_i(\y') \leq t'(f_i(\y) - f_i(\y')); \label{eq:const2}\\
  p_i(\z') - p_i(\z) \leq t(f_i(\z') - f_i(\z)); \label{eq:const3}\\
  p_i(\z) - p_i(\z') \leq t'(f_i(\z) - f_i(\z')). \label{eq:const4}
 \end{align}
 Suppose that $t < t'$ (the case $t > t'$ is symmetric and omitted). If $f_i(\y) < f_i(\y')$, then it follows that \eqref{eq:const1} and \eqref{eq:const2} cannot be both satisfied; if instead $f_i(\y) > f_i(\y')$, then \eqref{eq:const3} and \eqref{eq:const4} cannot be both satisfied.

The equality of payments directly follows from the outcomes being the same, and the profiles being separated.
\end{proof}
Lemma~\ref{lem:strongly} thus states that if in some path $P$ of $\T$ the $(k+2)$-th query to $i$ is ineffective, then it must be strongly ineffective. Similarly, if in some path $P$ of $\T$ the $(k+2)$-th query to $i$ is only-maximum (only-minimum, resp.) effective, and the types of $i$ different from the maximum (minimum, resp.) are separated at the $(k+2)$-th query or successive, then the $(k+2)$-th query is strongly only-maximum (only-minimum, resp.) effective. In order to stress that in $\T$ this property must be satisfied, we say that the corresponding mechanism is \emph{strongly almost $k$-limited}.

Lemma~\ref{lem:ineffective} and Lemma~\ref{lem:strongly} allow to conclude that a $k$-step OSP mechanism must be strongly almost $k$-limited. However, in such a mechanism we are still allowed to query agent $i$ more than $k+2$ times, or to have a $(k+2)$-th query that is neither a revelation nor an extremal query.  We next show how to transform $\T$ in order to achieve the desired reduction to a $k$-limited mechanism.
\begin{lemma}[Revealing at the $(k+2)$-th query/1]
\label{lem:transform1}
 Let $\M$ be a $k$-step OSP mechanism that implements $(f,p)$ with implementation tree $\T$, and consider a path $P$ of $\T$ such that along $P$ agent $i$ is queried at least $k+2$ times, and the $(k+2)$-th query is strongly ineffective. Then, there is a $k$-step OSP mechanism $\M'$ that implements $(f,p)$ with an implementation tree $\T'$ such that along path $P$, agent $i$ is queried exactly $k+2$ times and the $(k+2)$-th query is a strongly ineffective revelation query.
\end{lemma}
\begin{proof}
Let $u$ be the node of $P$ corresponding to the $(k+2)$-th query to $i$ along $P$, and for each child $v$ of $u$, let $\T_v$ be the subtree of $\T$ rooted at $v$. We define the implementation tree $\T'$ obtained from $\T$ by replacing each child $v$ of $u$ with nodes $v_t$ for each type $t \in D_i(v)$, each of them rooting a copy of $\T_v$ (pruned of the redundant queries to $i$ for types $t' \neq t$). 

 It is immediate to check that the mechanism $\M'$ implemented by $\T'$ implements $(f, p)$, and, along $P$, agent $i$ is queried $k+2$ times, with the $(k+2)$-th query being a strongly ineffective revelation query. We next show that $\M'$ is $k$-step OSP. To this aim, we need to prove that all the $k$-step OSP constraints are satisfied. As for those constraints defined by nodes $u'$ such that $i(u') \neq i$, it is immediate to check that they are trivially satisfied, since they are satisfied also in the original mechanism and no new separation has been introduced 
 for $i(u')$.

 Similarly, the validity of all constraints about profiles that have been separated in $\T'$ by $i$ at a node different from $u$ is inherited from the original mechanism. It is only left to check that the constraints introduced at node $u$ are satisfied. Since the new queries at $u$ in $\T'$ are for singletons, then the new constraints require that for every child $v$ of $u$ in $\T$ such that $|D_i(v)| \geq 2$, every $t, t' \in D_i(v)$, and every $\xii$ and $\yi$ available at $u$, it holds that
\begin{align}
	p_i(t, \xii) - p_i(t', \yi) & \leq t'(f_i(t, \xii) - f_i(t', \yi)) \label{eq:ospTnew1}\\
	p_i(t', \yi) - p_i(t, \xii) & \leq t(f_i(t', \yi) - f_i(t, \xii))\label{eq:ospTnew2}
\end{align}
 However, since the query at $u$ is strongly ineffective, then it must be the case that
 $f_i(t, \xii) = f_i(t', \yi)$ and $p_i(t, \xii) = p_i(t', \yi)$, and thus these constraints are clearly satisfied.
\end{proof}

\begin{lemma}[Revealing at the $(k+2)$-th query/2]
\label{lem:transform2}
 Let $\M$ be a $k$-step OSP mechanism that implements $(f,p)$ with implementation tree $\T$, and consider a path $P$ of $\T$ such that along $P$ agent $i$ is queried at least $k+2$ times, and the $(k+2)$-th query is strongly only-maximum (only-minimum) effective. Then, there is a $k$-step OSP mechanism $\M'$ that implements $(f,p)$ with an implementation tree $\T'$ such that along path $P$, agent $i$ is queried exactly $k+2$ times and the $(k+2)$-th query is a strongly only-maximum (only-minimum, resp.) revelation query.
\end{lemma}
\begin{proof}
The transformation and the proof are exactly the same as for Lemma~\ref{lem:transform1}, except that we need to check that the constraints \eqref{eq:ospTnew1} and \eqref{eq:ospTnew2} 
 are satisfied even if $t' = \max \{t \in D_i(u)\}$ ($t' = \min \{t \in D_i(u)\}$, resp.) and $f_i(t', \yi) \neq f_i(t, \xii)$.

 However, since the query at $u$ is strongly only-maximal (only-minimal, resp.) effective, then $f_i(t, \yi) = f_i(t, \xii) \neq f_i(t', \yi)$. However, since $(t, \yi)$ and $(t', \yi)$ differ only in the type of $i$, and they have a different outcome, then they must be separated by $i$. Since $\M$ is $k$-step OSP, then the following $k$-step OSP constraints hold:
 \begin{align}
  p_i(t, \yi) - p_i(t', \yi) & \leq t'(f_i(t, \yi) - f_i(t', \yi)); \label{eq:ospT1}\\
  p_i(t', \yi) - p_i(t, \yi) & \leq t(f_i(t', \yi) - f_i(t, \yi)). \label{eq:ospT2}
 \end{align}
 However, since the query at $u$ is strongly only maximal (only-minimal, resp.), then $p_i(t, \xii) = p_i(t, \yi)$ and $f_i(t, \xii) = f_i(t, \yi)$. Then \eqref{eq:ospTnew1} and \eqref{eq:ospTnew2} follow  from \eqref{eq:ospT1} and \eqref{eq:ospT2}, respectively.
\end{proof}

Theorem~\ref{thm:klimited} then follows from the repeated applications of transformations described in Lemma~\ref{lem:transform1} and Lemma~\ref{lem:transform2}, since for each remaining path $P$ it must be the case that either agent $i$ is queried at most $k+1$ times, or it is queried $k+2$ times, with the $(k+2)$-th query either being a revelation or an extremal query. 

\section{Cycle-Monotonicity for $k$-step OSP}
We will next show that $k$-step OSP-ness can be stated in terms of suitable weighted graphs and their cycles.

\subsection{Definition of Cycle-Monotonicity}
For each agent $i$, for each path $P$ from the root of $\T$ to a node $u$ corresponding to the $(k+2)$-th query to $i$, if it exists, or to a leaf otherwise, we partition the domain $D_i$ of types of $i$ in classes $D_{i, P}^1, \ldots, D_{i, P}^{\ell}$, where $\ell = \min\{k+2, q+1\}$, $q$ denoting the number of queries to $i$ along this path; $D_{i, P}^j$, for $j < \ell$ contains all types of $i$ available at the $j$-th query at $i$, but not available at the $(j+1)$-th query to $i$ along $P$, and $D_{i, P}^{\ell} = D_i \setminus \bigcup_{j = 1}^{\ell-1} D_{i, P}^j$.

Suppose that, along the path $P$, $i$ is queried at least $k+2$ times, and let $u$ be the node corresponding to the $(k+2)$-th query at $i$. Let $D_{i, P}^{k+2, \notE}$ be a maximal set of types of $i$ available at $u$ that receive the same outcome for each possible fixed profile of the remaining agents available at $u$, i.e., $D_{i, P}^{k+2, \notE} \subseteq D_{i, P}^{k+2}$ such that for every $t, t' \in D_{i, P}^{k+2, \notE}$ and every $\xii$ available at $u$ we have that $f_i(t, \xii) = f(t', \xii)$, and for every $t \in D_{i, P}^{k+2, \notE}$ and every $t'' \in D_{i, P}^{k+2, E} = D_{i, P}^{k+2} \setminus D_{i, P}^{k+2, \notE}$ there is $\yi$ available at $u$ for which $f_i(t, \yi) \neq f_i(t'', \yi)$. Note that if the $(k+2)$-th query is (strongly) ineffective, we have that $D_{i, P}^{k+2, \notE} = D_{i, P}^{k+2}$. Moreover, if the $(k+2)$-th query is (strongly) only-maximum or only-minimum effective, then $D_{i, P}^{k+2, E}$ contains only the one extremal type for which the $(k+2)$-th query is effective.

Given this partition of the type domain of agent $i$, we then partition the profiles in equivalence classes as follows: for each agent $i$, for each path $P$ from the root of $\T$ to a node $u$ corresponding to the $(k+2)$-th query to $i$, if it exists, or to a leaf otherwise, we define the equivalence classes 
\begin{align}\label{eq:Lambdas}
\Lambda_{i, P}^j(0) & = \{(x_i, \xii) \mid x_i \in D_{i, P}^j, \xii \text{ available at } u, f_i(x_i, \xii) = 0\}\\ \Lambda_{i, P}^j(1) & = \{(x_i, \xii) \mid x_i \in D_{i, P}^j, \xii \text{ available at } u, f_i(x_i, \xii) = 1\},
\end{align}
 for $j = q+1$ if in $P$ agent $i$ is queried at most $q \leq k+1$ times, and $j \in \{(k+2, E), (k+2, \notE)\}$ otherwise.
%
Let $\mathbf{\Lambda}_i$ be the set that contains all the equivalence classes defined for agent $i$.
Moreover, we abuse notation and we set $f_i(\Lambda_{i, P}^j(0)) = 0$ and $f_i(\Lambda_{i, P}^j(1)) = 1$. We are now ready to define the graphs of interest for $k$-step OSP. 

\begin{definition}(\vgraph)
	Let $f$ be a social choice function and $\T$ be an implementation tree. We define for every agent $i$, the {\em \vgraph} $\verk$ with $\mathbf{\Lambda_i}$ as the set of vertices, and an edge $e$ between $\Lambda, \Lambda' \in \mathbf{\Lambda_i}$ exists if there are $\x \in \Lambda$ and $\x' \in \Lambda'$ such that $\x$ and $\x'$ have been separated in $\T$ by $i$. We set $w(e)= \min_{x_i \mid \exists \x_{-i} \colon (x_i, \x_{-i}) \in \Lambda} x_i (f_i(\Lambda') - f_i(\Lambda))$.
\end{definition}

We say that the $k$-step OSP cycle monotonicity ($k$-step OSP CMON) property holds  if, for all $i$, the graph $\verk$ does not have negative weight cycles.

\subsection{Proving $k$-step OSP CMON}
We begin by observing that if the mechanism is $k$-limited, the equivalence classes \eqref{eq:Lambdas} can be used for redefine some of the $k$-step OSP constraints given in \eqref{eq:osp_constraint} as follows. Towards this end, we define the function $\mathbf{\lambda}_i^{-1}$ that on input a profile $\x$ returns the equivalence class $\Lambda \in \mathbf{\Lambda}_i$ to which $\x$ belongs. 
\begin{lemma}[Rewriting the $k$-Step OSP Constraints]
\label{lem:osp_constraint_new}
 If a $k$-limited mechanism $\M$ with implementation tree $\T$ is $k$-step OSP, then for all $i$, all vertices $u$ of $\T$ such that $i=i(u)$, and every pair of profiles $\a$ and $\b$ separated at $u$ such that $\lambda^{-1}(\a) \neq \lambda^{-1}(\b)$, the following holds:
    \begin{equation}
     \label{eq:osp_constraint_new}
     p_i(\b) - p_i(\a) \leq \min_{c \mid \exists \ci \colon (c, \ci) \in \lambda^{-1}_i(\a)} c\left(f_i(\b)-f_i(\a)\right).
    \end{equation}
\end{lemma}
\begin{proof}
Let $\tilde{P}$ be the path from the root of $\T$ to the leaf $\ell$ corresponding to profile $\a$.
Assume first that agent $i$ has been queried $q \leq k+1$ times in $\tilde{P}$. This implies that $\lambda^{-1}_i(\a) = \Lambda^{q+1}_{i,\tilde{P}}(\a)$ and that $\Gamma^k_u(\a)$ contains all profiles available at $\ell$. In turns, $(t, \ai) \in \Gamma^k_u(\a)$ f or every $t \in D^{q+1}_{i, \tilde{P}}$, and no profile $\x$ such that $x_i \notin D^{q+1}_{i, \tilde{P}}$ is in $\Gamma^k_u(\a)$ since $x_i$ has been separated from $a_i$ in one of the $q$ queries to $i$. Note that $f_i(t, \ai) = f_i(\a)$ for every $t \in D^{q+1}_{i, \tilde{P}}$, since these profiles are never separated by any agent. Observe, moreover, that $\Lambda^{q+1}_{i,\tilde{P}}(\a)$ contains all profiles available at $\ell$ with outcome $f_i(\a)$. Hence, as discussed above, it must contain at least $(t, \ai)$ for every $t \in D^{q+1}_{i, \tilde{P}}$. Hence, from Lemma~\ref{lem:def}, we have that
\begin{align*}
 p_i(\b) - p_i(\a) & \leq \min_{c \mid \exists \ci \colon (c, \ci) \in \Gamma^k_u(\a)} c\left(f_i(\b)-f_i(\a)\right)\\
 & = \min_{c \in D^{q+1}_{i, \tilde{P}}} c\left(f_i(\b)-f_i(\a)\right)\\
 & = \min_{c \mid \exists \ci \colon (c, \ci) \in \lambda^{-1}(\a)} c\left(f_i(\b)-f_i(\a)\right).
\end{align*}

Consider now the case in which agent $i$ is queried at least $k+2$ times $\tilde{P}$. Let $P$ denote the subpath of $\tilde{P}$ containing all the nodes from the root to the node $v$ in which $i$ receives the $(k+2)$-th query. Note that $a_i \in D_{i, P}^{k+2}$.
Suppose that $f_i(\b) \geq f_i(\a)$ (the proof for $f_i(\b) < f_i(\a)$ is symmetric and hence omitted). Hence, the minimum in the r.h.s. of \eqref{eq:osp_constraint_new} is achieved in correspondence of $c^* = \min \{c \mid \exists \ci \colon \c = (c, \ci) \in \lambda^{-1}_i(\a)\}$.

We next show that since $\M$ is $k$-limited then $\c^* = (c^*, \ai)$ also belongs to $\lambda^{-1}_i(\a)$. Since, for some $\ci$, $(c^*, \ci) \in \lambda^{-1}_i(\a)$, then (i) $f_i(c^*, \ci) = f_i(\a)$ and (ii) $c^*$ has not been separated by $a_i$ within the first $k+1$ queries to $i$. Thus, we have that $c^* \in D_{i, P}^{k+2}$. Since $\M$ is $k$-limited, then the $(k+2)$-th query is either strongly ineffective, and thus $a_i, c^* \in D_{i, P}^{k+2, \notE}$, or (strongly) only-maximum or only-minimum effective, and thus $D_{i, P}^{k+2, E}$ is a singleton, and hence either $c^* = a_i \in D_{i, P}^{k+2, E}$ or $a_i, c^* \in D_{i, P}^{k+2, \notE}$. The claim then follows since $\ai$ is available at node $v$ (since it is still available at the leaf $\ell$ in the subtree rooted at $v$).

We now argue that $b_i$ does not belong to the same partition of $a_i$ and $c^*$ at the history in which the $(k+2)$-th query is performed. Since the mechanism is $k$-limited and $\lambda^{-1}(\a) \neq \lambda^{-1}(\b)$, then either $\a$ and $\b$ have been separated by $i$ in the first $k+1$ queries (i.e., $b_i \not\in  D_{i, P}^{k+2}$) or they have been separated in one of the following queries to $i$, and either $a_i= c^* \in D_{i, P}^{k+2, E}$ and $b_i \in D_{i, P}^{k+2, \notE}$, or $a_i, c^* \in D_{i, P}^{k+2, \notE}$ and $b_i \in D_{i, P}^{k+2, E}$.

Thus, since the mechanism is $k$-limited, then in all cases we have that $\c^*$ and $\b$ also have been separated at the same node $u'$ in which $\a$ and $\b$ have been separated.
Hence, we have the following $k$-step OSP constraints corresponding to $\c^*$ and $\b$:
\begin{equation}
\label{eq:constraint_rep}
 p_i(\b) - p_i(\c^*) \leq \min_{c \mid \exists \ci \colon (c, \ci) \in \Gamma^k_{u'}(\c^*)} c\left(f_i(\b)-f_i(\c^*)\right).
\end{equation}
However, since the mechanism is $k$-limited and $\c^* \in \lambda^{-1}(\a)$, we have that $f_i(\a) = f_i(\c^*)$ and $p_i(\a) = p_i(\c^*)$. Hence, \eqref{eq:constraint_rep} can be rewritten as follows:
\begin{equation}
\label{eq:constraint_rep2}
 p_i(\b) - p_i(\a) \leq \min_{c \mid \exists \ci \colon (c, \ci) \in \Gamma^k_{u'}(\c^*)} c\left(f_i(\b)-f_i(\a)\right).
\end{equation}
As above, the minimum in the r.h.s. of \eqref{eq:constraint_rep2} is achieved for $c' = \min \{c \mid \exists \ci \colon (c, \ci) \in \Gamma^k_{u'}(\c^*)\}$. Since $\c^* \in \Gamma^k_{u'}(\c^*)$, then either $c' = c^*$ or $c' < c^*$ implying \eqref{eq:osp_constraint_new}. 
\end{proof}



Secondly, we note that the {\vgraph}s of $k$-limited mechanisms enjoy the following useful property.
\begin{lemma}[Sticky Edges]
 \label{lem:sep_profiles}
 Suppose that $\verk$ is the {\vgraph} of a $k$-limited mechanism and there is an edge $(\Lambda, \Lambda')$ in $\verk$, then for every $\x \in \Lambda$, and every $\x' \in \Lambda'$, we have that $\x$ and $\x'$ have been separated by $i$ at the same node in $\T$.
%
\end{lemma}
\begin{proof}
 Since the edge $(\Lambda, \Lambda')$ exists in $\verk$, then there must be $\y \in \Lambda$ and $\y' \in \Lambda'$ such that $\y$ and $\y'$ have been separated at some node $u$ by $i$.

 Note that for every $\x \in \Lambda$ and $\x' \in \Lambda'$, it is impossible that both $x_i \in D_{i, P}^{k+2, \notE}$ and $x'_i \in D_{i, P}^{k+2, \notE}$, since in a $k$-limited mechanism if these two profiles had been separated by $i$ at the $(k+2)$-th query, then this must be either strongly ineffective or (strongly) only-maximum or only-minimum effective. However, in all these cases $f_i(\x) = f_i(\x')$, and thus $\x$ and $\x'$ cannot belong to different equivalence classes.

 Hence, since the mechanism is $k$-limited, then either $u$ is one of the first $k+1$ nodes in which $i$ is queried, or it is a (strongly) only-maximum (only-minimum, resp.) effective $(k+2)$-th query and the largest (smallest, resp.) among $y_i$ and $y'_i$ belongs to $D_{i, P}^{k+2, E}$, while the other is in $D_{i, P}^{k+2, \notE}$.

 Moreover, by definition of $\Lambda$, for every $\x \in \Lambda$, $\xii$ is still available at $u$ and $x_i$ has not been separated from $y_i$ by $i$ in the first $k+1$ queries, and if they have been separated at the $(k+2)$-th query, then this query is strongly only-maximal (only-minimal, resp.) effective and either $x_i = y_i$ or neither $x_i$ nor $y_i$ are the extreme for which this query is effective. The same holds for $\x'$ with respect to $\y'$. Hence, $\x$ and $\x'$ must have been separated at $u$.
\end{proof}

We then have the following theorem.
\begin{theorem}\label{thm:rev_cmon}
    A mechanism $\M$ with implementation tree $\T$ is $k$-step OSP for a social function $f$ on finite domains if and only if it is $k$-limited and $k$-step OSP CMON holds.
\end{theorem}
\begin{proof}
 If $\M$ is $k$-step OSP, then, by Theorem~\ref{thm:klimited}, it is $k$-limited.
 Now, suppose that there is a negative cycle in $\verk$.
	Let the negative-weight cycle be $C = (\Lambda^1, \Lambda^2,  \ldots, \Lambda^{h+1})$, with $\Lambda^{h+1}= \Lambda^1$.
%
%
By definition of \vgraph\ there are profiles $\a^j \in \Lambda^j$, for $j=1, \ldots, h$. By Lemma \ref{lem:sep_profiles}, 
each profile $\a^j$ is separated by $i$ from $\a^{j+1}$. Since $\lambda^{-1}(\a^{j}) \neq \lambda^{-1}(\a^{j+1})$, for each $j=1, \ldots, h$, then 
Lemma~\ref{lem:osp_constraint_new} implies that the following $k$-step OSP constraints must be satisfied:
{\allowdisplaybreaks
	\begin{align*}
		p_i(\a^2) - p_i{(\a^1)} & \leq \min_{c \mid \exists\ci \colon (c, \ci) \in \Lambda^1} c (f(\a^2) - f(\a^1)), \\
		p_i(\a^3) - p_i{(\a^2)} & \leq \min_{c \mid \exists\ci \colon (c, \ci) \in \Lambda^2} c(f(\a^3) - f(\a^2)), \\
		& \quad \cdots \\
		p_i(\a^h) - p_i{(\a^{h-1})} & \leq  \min_{c \mid \exists\ci \colon (c, \ci) \in \Lambda^{h-1}} c(f(\a^{h}) - f(\a^{h-1})), \\
		p_i(\a^1) - p_i{(\a^{h})} & \leq  \min_{c \mid \exists\ci \colon (c, \ci) \in \Lambda^k} c(f(\a^{1}) - f(\a^h)).
	\end{align*}}
	Since all above inequalities are satisfied, then it must be the case that the inequality achieved by summing up all above inequalities must be satisfied. However, observe that by summing up the left-hand side we achieve $0$ since each $p_i(\x)$ is added and subtracted once. Instead, the terms on the right-hand side are exactly the cost of edges of $C$, and thus they sum exactly to $w(C)$. Hence, we obtain the contradiction that $0 \leq w(C) < 0$.

 Consider now the case that $\M$ is $k$-limited and there are no negative cycles in $\verk$ for every $i$. We then show that there are payments $p$ such that $\M=(f, p)$ is $k$-step OSP with implementation tree $\T$.
	Fix $i$ and $\T$ and consider the graph $\verk$. 
%
	Augment $\verk$ with a
	node $\omega$ and edges $(\omega, \Lambda)$ for any $\Lambda \in \mathbf{\Lambda}_i$ each of
	weight $0$. Observe that $\omega$ does not belong to any
	cycle of the augmented $\verk$ since it has outgoing edges only.
	Therefore we can focus our attention on cycles of $\verk$. For any $\Lambda \in \mathbf{\Lambda}_i$, we let $\shp(\omega, \Lambda)$ be the length of the shortest path from $\omega$ to $\Lambda$ in the
	augmented \vgraph. Since $\verk$ is finite and does not have
	negative-weight cycles then $\shp(\omega,\Lambda)$ is well defined. It suffices, for
	all $\Lambda \in \mathbf{\Lambda}_i$, to set $p_i(\x) = \shp(\omega, \Lambda)$ for every $\x \in \Lambda$.
	Indeed, consider two profiles $\a$ and $\b$ separated by $i$ at $u$, and a profile $\c \in \Gamma_u^k(\a)$. If $\a$ and $\b$ belong to the same class $\Lambda$, then $p_i(\b) - p_i(\a) = 0$ and $f_i(\b) - f_i(\a) = 0$, and thus the $k$-step OSP constraint \eqref{eq:osp_constraint} is satisfied. If $\a$ and $\b$ belong to different classes $\Lambda$ and $\Lambda'$, then the edge $(\Lambda,\Lambda')$ in $\verk$. The fact that the shortest path from $\omega$ to $\Lambda$ followed by the edge $(\Lambda,\Lambda')$ is a path from $\omega$ to $\Lambda'$ implies, by shortest path definition, that $p_i(\b) - p_i(\a) = \shp(\omega,\Lambda') - \shp(\omega,\Lambda) \leq c_i(f_i(\b) - f(\a))$, and thus the $k$-step OSP constraint \eqref{eq:osp_constraint} is satisfied. 
\end{proof}

\section{Algorithmic Characterization}\label{sec:algo}
We next provide the algorithmic characterization for $k$-step OSP mechanisms. We will start by recalling the characterization for $k = \infty$. We will see that how this characterization can be easily extended to each $k \geq 0$.
%
As for $k = \infty$, the characterization has been provided by Ferraioli et al. \cite{wine21}. We will report here this result, since it will  be useful in what follows. To this aim, we say that an agent $i$ is \emph{revealable} at node $u$ under social choice function $f$ if either $f_i(\x) = 1$ for every $\x$ such that $x_i \in D_i(u)$ and $x_i < \max D_i(u)$, and $\x_{-i}$ is compatible with $u$, or $f_i(\x) = 0$ for every $\x$ such that $x_i \in D_i(u)$ and $x_i > \min D_i(u)$, and $\x_{-i}$ is compatible with $u$.

An algorithm $f \colon \times_{i = 1}^n D_i \rightarrow \out$ is \emph{two-way greedy implementable} if it can be implemented by a round-table mechanism with implementation tree $\T$ such that:
\begin{itemize}[nosep]
 \item at each node $u$ the agent $i(u)$ separates her domain $D_i(u)$ in two ordered subsets $L_i(u)$ and $R_i(u)$, i.e., $\max \{t \in L_i(u)\} < \min \{t \in R_i(u)\}$, such that at least one of the two is a singleton, i.e., it contains a single type, being the minimum in $D_i(u)$ (if $L_i(u)$ is a singleton) or the maximum (if $R_i(u)$ is a singleton). If $L_i(u)$ is a singleton, then the agent $i(u)$ receives outcome $1$ in every profile compatible with $L_i(u) \times D_{-i}(u)$, and we say that $i(u)$ interacts with the mechanism in a greedy fashion. If $R_i(u)$ is a singleton, then the agent $i(u)$ receives outcome $0$ in every profile compatible with $R_i(u) \times D_{-i}(u)$, and we say that $i(u)$ interacts with the mechanism in a reverse greedy fashion;
 \item each agent $i$ is not allowed to interleave interaction in a greedy fashion with interaction in a reverse greedy fashion until it is revealable. Specifically, agent $i$ at node $u$ is either revealable, or can interact with the mechanism in a greedy fashion if and only it $i$ interacted with the mechanism in a greedy fashion at every node $u'$ along the path between the root of $\T$ and $u$ such that $i = i(u')$.
\end{itemize}
Then we have the following theorem.
\begin{theorem}[\cite{wine21}]
\label{thm:wine}
There is an OSP mechanism implementing $f$ if and only $f$ is two-way greedy implementable.
\end{theorem}
While the definition of two-way greedy implementable mechanism appears to be quite involved, it has been observed that it establishes a very strong relationship between two-way greedy implementable algorithms and greedy algorithms \cite{wine21}. Essentially every greedy algorithm and every reverse greedy algorithm (a.k.a.,  deferred-acceptance algorithm) is two-way greedy, and hence can be turned in an OSP mechanism. This continues to hold if the algorithm is allowed to greedily insert into the solutions some components (i.e., interact greedily with some agents) and reverse greedily remove from the solutions other components (i.e., interact reverse greedily with other agents). We will finally observe that whenever an agent is revealable at some node $u$, we can ask the agent to reveal her type without affecting the OSPness of the mechanism.

\subsection{$k$-step OSP-graph vs. $\infty$-step OSP-graph}
We first prove a useful relation between negative cycles in the $k$-step OSP-graph, and negative cycles in the $\infty$-step OSP-graph. This will be useful to provide later a characterization of $k$-step OSP mechanisms as two-way greedy mechanisms enjoying a special adjunctive feature.
\begin{lemma}
\label{lem:inf_vs_k}
 Let $\M$ be a {$k$-limited} mechanism with implementation tree $\T$. 
 There is no negative cycle in $\verk$ if and only if there is no negative cycle  in $\verinf$.
\end{lemma}
\begin{proof}
 Suppose that there is a negative cycle $C$ in $\verinf$. Now consider the sequence of nodes $C'$ in $\verk$ defined as follows: for each profile $\x \in C$, we insert in $C'$ the node $\Lambda^j$ such that $\x \in \Lambda^j$. Now for each edge $(\x, \y)$ of $C$ such that $\x \in \Lambda^j$ and $\y \in \Lambda^h$, for $j \neq h$, we must have that edge $(\Lambda^j, \Lambda^h)$ exists, and
 $$w(\Lambda^j, \Lambda^h) = \min_{t \mid \exists \ci \colon (t, \ci) \in \Lambda^j} t(f_i(\Lambda^h) - f_i(\Lambda^j)) \leq x_i(f_i(\y) - f_i(\x))=w(\x, \y). $$
 Moreover, for each edge $(\x, \y)$ of $C$ such that $\x, \y \in \Lambda^j$, we have that $w(\x, \y)=x_i(f_i(\y) - f_i(\x)) = 0$. Hence, by pruning from $C'$ the possible consecutive replications of the same node, we get a cycle in $\verk$ whose cost is at most the cost of $C$, and thus is negative, as desired.

 Consider now the case that there is a negative cycle $C$ in $\verk$. We can suppose w.l.o.g. that $C$ is a simple cycle (i.e., it visits each node only once), otherwise, in order for $C$ to be negative, there must be at least one negative simple cycle within $C$, that we can take in place of $C$. We consider the sequence $C'$ of nodes of $\verinf$ defined as follows: for each $\Lambda^j \in C$, we add to $C'$ the profile $\x^j = (x^j_i, \x^j_{-i})$ such that $x^j_i = \arg \min_{t \mid \exists  \ci \colon (t,  \ci) \in \Lambda^j}$ and $\x^j_{-i}$ be one of the profiles of the other agents' actions for which  $\x^j \in \Lambda^j$. 
 Lemma~\ref{lem:sep_profiles} implies that edge $(\x^j, \x^{j+1})$ exists in $\verinf$ and has cost $x^j_i (f_i(\x^{j+1}) - f_i(\x^j)) = x^j_i (f_i(\Lambda^{j+1}) - f_i(\Lambda^j))$, that, by our choice of $x_i^j$, is exactly the cost of edge $(\Lambda^j, \Lambda^{j+1})$ in $C$. Hence, we can conclude that the cost of $C'$ is exactly the cost of $C$, and thus negative, as desired. 
\end{proof}

\subsection{$k$-step OSP Characterization}
Two-way greedy implementation of a social choice function $f$ assumes that the implementation tree $\T$ is binary and  makes only extremal queries. Clearly, each implementation tree $\T'$ for which there are no negative cycles in $\mathcal{O}^{\infty}_{i, f, \T'}$, regardless the kind of queries that are performed in $\T'$, can be transformed in an implementation tree $\T$ with binary extremal queries that still has no negative cycle through a \emph{serialization} procedure (see \cite[Observation~3 and Theorem~4]{wine21}). Similarly, a binary implementation tree $\T$ with extremal queries and no negative cycle in $\verinf$ can be transformed in an implementation tree $\T'$ with generic queries but still no negative cycle through a simple \emph{compression} procedure: every two consecutive nodes $u$ and $u'$ such that $i = i(u) = i(u')$ can be merged in a single node $U$ with outgoing edges leading to $v_1, \ldots, v_x$, and to $v'_1, \ldots, v'_y$, where $v_1, \ldots, v_x$ are the children of $u$ different from $u'$, and $v'_1, \ldots, v'_y$ are the children of $u'$. However, these serialization/compression operations will change the number of queries done to each agent. While this does not matter for $\infty$-step OSPness, it turns out to be very relevant for $k$-step OSPness. Hence, we will introduce a way to keep the number of queries to a given agent unchanged after the operations of serialization/compression. Specifically, we say that a binary implementation tree $\T$ with extremal queries is \emph{$k$-limitable} if and only if the implementation tree $\T'$ achieved through the compression procedure described above is \emph{$k$-limited}. (Note that for convenience we are abusing a bit our terminology by calling $k$-limited the implementation tree $\T'$ rather than the mechanism $\M$ using $\T'$.)

We then have the following characterization.
\begin{theorem}
\label{thm:kgreedy}
	There exists a $k$-step OSP mechanism implementing $f$ if and only $f$ is two-way greedy implementable and the corresponding implementation tree $\T$ is $k$-limitable.
\end{theorem}
\begin{proof}
 Suppose that there exists a $k$-step OSP mechanism implementing $f$ with implementation tree $\T$.
 By Theorem~\ref{thm:rev_cmon}, $\T$ is $k$-limited and $k$-step OSP CMON holds. Note that since $k$-step OSP CMON holds, then, by Lemma~\ref{lem:inf_vs_k}, there is no negative cycle in $\verinf$.
 Consider the implementation tree $\T'$ achieved by running the serialization procedure described above on $\T$.
 Since $\T$ is $k$-limited, then it follows that $\T'$ is $k$-limitable. Moreover, since there is no negative cycle in $\verinf$, then there is no negative cycle in $\mathcal{O}^{\infty}_{i, f, \T'}$. Hence, by Theorem~\ref{thm:wine} $f$ is two-way greedy. 
 

 Suppose now that $f$ is two-way greedy implementable through a $k$-limitable implementation tree $\T$.
 Then, by Theorem~\ref{thm:wine}, there is no negative cycle in $\verinf$. Consider then the implementation tree $\T'$ achieved by running the compression procedure on $\T$; there is no negative cycle in $\mathcal{O}^{\infty}_{i, f, \T'}$. Moreover, since $\T$ is $k$-limitable, then $\T'$ is $k$-limited and thus, by Lemma~\ref{lem:inf_vs_k}, $\mathcal{O}^{k}_{i, f, \T'}$ has no negative cycles. Then, by Theorem~\ref{thm:rev_cmon}, we have that the corresponding mechanism is $k$-step OSP.
\end{proof}

Essentially, Theorem~\ref{thm:kgreedy} states that the characterization in term of greedy algorithms provided for OSP continues to hold even for $k$-step OSPness. However, for these mechanism we also require a further constraint to be satisfied, namely that the implementation tree is $k$-limited (in its compressed version) or $k$-limitable (in the serialized version). This essentially means that the mechanism can interact with each agent at most $k+2$ times (with the $(k+2)$-th interaction limited as discussed above).

The effect of this limitation is very heavy in the case of strong OSP mechanisms. Indeed, a $0$-limitable two-way greedy implementable $f$ can be implemented by a compressed implementation tree $\T$ that can be described as follows. For each agent $i$, for each path $P$ from the root to a leaf, let $u$ be the first node in $P$ such that $i = i(u)$ and let $v_1, \ldots, v_\ell$ be the  children of $u$ with $\max D_i(v_j) < \min D_i(v_{j+1})$; then for each $\x_{-i}$ compatible with $u$ one of the following three cases occurs:
\begin{itemize}[nosep]
 \item there is $j^* \in \{0, 1, \ldots, \ell, \ell+1\}$ such that $f_i(x_i, \x_{-i}) = 1$ for $x_i \in \bigcup_{j = 1}^{j^*} D_i(v_j)$, and $f_i(x_i, \x_{-i}) = 0$ for $x_i \in \bigcup_{j = j^*}^{\ell} D_i(v_j)$;
 \item $f_i(x_i, \x_{-i}) = 0$ for $x_i \in \bigcup_{j = 2}^{\ell} D_i(v_j)$, $f_i(x_i, \x_{-i}) = 1$ for $x_i \in D_i(v_1)$ with $x_i < \max D_i(v_1)$, and $f_i(x^*_i, \x_{-i}) \in \{0, 1\}$ for $x^*_i = \max D_i(v_1)$;
 \item $f_i(x_i, \x_{-i}) = 1$ for $x_i \in \bigcup_{j = 1}^{\ell-1} D_i(v_j)$, $f_i(x_i, \x_{-i}) = 0$ for $x_i \in D_i(v_\ell)$ with $x_i > \min D_i(v_\ell)$, and $f_i(x^*_i, \x_{-i}) \in \{0, 1\}$ for $x^*_i = \min D_i(v_\ell)$.
\end{itemize}
That is, either the separation between outcomes $0$ and $1$ is decided at the first query, or we are allowed to run a second query affecting the outcome of agent $i$ only for separating the maximum of the group of minimal types or the minimum of the group of maximal types.

This can be easily characterized to every $k$, by essentially stating that a social function $f$ can be implemented by a $k$-step OSP mechanism only if the possible thresholds for the outcomes can be determined in at most $k+1$ compressed queries or at most $k+2$ non-interleaving queries.

%% file: apps.tex
In this section we apply our characterization to quantify the restriction that $k$-step OSP imposes on the quality of the algorithmic solution that can be implemented. We will focus on maximization problems; agents' types will thus be a non-negative valuation (i.e., a non-positive cost) for each algorithmic allocation received. 

To introduce our questions of interest, we start by discussing arguably the simplest possible problem in this area, \emph{social-welfare maximizing single-item auctions}: $n$ agents have a valuation $v_i \in D_i$ for the item, and we are willing to sell the item to the agent with the highest valuation. Note that in order to fully define the social function $f$ for this problem, we need to specify how ties are broken. We will assume  that ties are broken in favor of the agent with the smallest index. We will also assume that every agent $i$ has a domain $D_i = D$. It is not hard to see how most of the arguments below can be adapted to work also if both these assumptions are dropped.

The ascending price auction discussed in Example \ref{example:auctions} is an OSP mechanism that solves this problem optimally. We can rephrase it in terms of two-way greedy implementation as follows: mark all agents as available; for $t$ from the smallest type to the second largest type in $D$ or until there is only one available agent, ask each available agent $i$ in order of their index whether her type is $t$, and in case of positive answer, mark the agent as unavailable; assign the item to the agent with smallest index among the available ones. Actually, this is not the unique OSP auction implementing the desired social function (e.g., we can query agents from the largest to the second smallest type and allocate the item upon a positive answer).

Our question is the following: is it possible to implement single-item auctions with a $k$-step OSP mechanism? And if not, how large can be the \emph{price of limited foresight}?
To answer these questions, first observe that the English auction described above queries each agent $|D|-1$ times, with the last query being a revelation strongly only-minimum effective query. Hence, this mechanism is $k$-limitable, and, by Theorem~\ref{thm:kgreedy}, $k$-step OSP for every $k \geq |D|-3$.

However, the mechanism fails to be $(|D|-4)$-limited. Indeed, after the $(|D|-3)$-th query to agent $i$, its domain contains the three largest types, and for each action of other agents compatible with the $(|D|-3)$-th query, the outcome of this agent should be the same when agent $i$ has the smallest and the second smallest type. But the ascending price auction described above does not provide such a guarantee.

Can there be another $k$-step OSP mechanism implementing the social choice function $f$? Or some social function $f'$ differing from $f$ only in the tie-breaking rule (and hence, still retuning an allocation that maximizes the social welfare)?

We will next show that a $k$-step OSP mechanism exists for this problem whenever $k \geq \left\lceil\frac{|D|}{2}\right\rceil - 2$ (e.g., an SOSP mechanism exists whenever $|D| \leq 4$). Actually, we prove that this result holds even in a more general setting than single-item auctions. We also prove that this result is tight.

\subsection{Weighted Matroids and $p$-systems}
We will now focus on the class of problems that satisfy the \emph{downward closed property}.
In these problems we are given a set $E$ of elements, with each element $e$ associated to a weight $w(e) \in \mathbb{R}$, and a set ${\cal F}$, containing subsets $S \subseteq E$, named \emph{feasible solutions}, that enjoy the following property: if $S \in {\cal F}$, then $T \in {\cal F}$ for every $T \subseteq S$,  and we need to select the one feasible solution of maximum total weight, i.e. $S^* = \arg \max_{S \in {\cal F}} \sum_{e \in S} w(e)$.

Most well-known problems belong to this class: the social welfare maximizing single-item auction can be seen as a problem in this class, where $E$ corresponds to the set of agents, their weight corresponds to their valuation for the item, and the set of feasible solution are all (and only) the singletons. Other examples of problems that can be easily modeled in this way are the ones of finding the maximum number of linearly independent rows in a matrix, or the one of finding the minimum spanning tree or the maximum independent set in a graph.

A special subclass of problems satisfying the downward closed property consists of problems defined on  (weighted) \emph{matroids}: for these problems, the set of feasible solutions also enjoys the \emph{exchange property}, that is, if $S \in {\cal F}$ and $T \in {\cal F}$ and $|S| \leq |T|$, then there is $e \in T \setminus S$ such $S \cup \{e\} \in {\cal F}$.

Given a subset $S$ of elements, we denote with $\mathcal{O}(S)$ the set of feasible solutions that are maximal with respect to $S$, i.e. $\mathcal{O}(S) = \{T \in \mathcal{F} \colon T \subseteq S, T \cup \{e\} \notin \mathcal{F} \; \forall \, e \in S \setminus T\}$. We also simply write $\mathcal{O}$ as shorthand for $\mathcal{O}(E)$, i.e., the maximal feasible solutions in $\mathcal{F}$.
Note that in a matroid, the exchange property implies that for every pair of solutions $S, T \in \mathcal{O}$, we have that $|S| = |T|$. For this reason, we can measure the distance between a generic problem satisfying the downward closed property and a problem defined on a matroid, by the extent to which two maximal feasible solutions differ in size. Specifically, given a downward-closed problem $(E, {\cal F})$, and a subset $S$ of $E$, the \emph{lower rank} of $S$ is $lr(S) = \min \{|T| \colon T \in {\cal O}(S)\}$, i.e., the size of the smallest maximal feasible solution with respect to $S$, and the \emph{upper rank} of $S$ is $ur(S) = \max \{|T| \colon T \in {\cal O}(S)\}$, i.e., the size of the largest maximal feasible solution with respect to $S$. Then, the distance of $(E, {\cal F})$ from being a matroid is defined as the \emph{rank quotient} $q(E, {\cal F}) = \min\left\{\frac{lr(S)}{ur(S)} \mid S \subseteq E, ur(S) \neq 0\right\}$, that is essentially an upper bound on the ratio between the sizes of the smallest and the largest maximal feasible solutions. We then say that, for $p \leq 1$, a downward-closed problem $(E, {\cal F})$ is a \emph{$p$-system} if $q(E, {\cal F}) = p$. Note that problems on matroids are $1$-systems.

Given a $p$-system $(E, {\cal F})$, let an algorithm process elements in some order $e_1, \ldots, e_n$ and let $E_j = \{e_1, \ldots, e_j\}$. If the algorithm returns a solution $G$ such that $G_j  = G \cap E_j$, with $j = 1, \ldots, n$, is a maximal feasible solution with respect to $E_j$, then the total weight of elements in $G$ is at least a fraction $p$ of the total weight of the optimal feasible solution \cite{Hausmann1980}. Recall that the \emph{greedy} algorithm processes elements $e$ in decreasing order of their weight, and includes them in the current solution $S$ unless $S \cup \{e\} \notin \mathcal{F}$, whilst the \emph{reverse greedy} algorithm processes elements $e$ in increasing order of their weight, and remove any solution $S$ containing $e$ from $\mathcal{O}$ unless this empties $\mathcal{O}$. It is not hard to see that both greedy and reserve greedy satisfy the property described above\footnote{While for greedy this is immediate (and indeed this is stated in \cite{Hausmann1980}), for reverse greedy, it follows by considering element $(e_1, \ldots, e_n)$ in decreasing order of weights (so that the algorithm will process them from the last to the first one), and observing there cannot be an element $e^* \in E_j \setminus G_j$, such that $G_j \cup \{e^*\}$ is feasible, since at the time the algorithm processes $e^*$, either no solution containing $e^*$ is removed from $\mathcal{O}$, and thus $G$ must contain $e^*$, or there is at least one solution in $\mathcal{O}$ that does not contain $e^*$ and it is maximal for $E$, and thus also for $E_j$, and it is returned by the algorithm.}, and hence they return a $p$-approximation of the optimal solution for every $p$-system.
This, in turn, means that for all these problems, it is possible to design an OSP mechanism that is able to return a $p$-approximate feasible solution for the problem \cite{wine21}.

We next show that it is always possible to implement the reverse greedy algorithm defined above as a $k$-limitable two-way greedy algorithm when $k \geq \left\lceil\frac{|D|}{2}\right\rceil - 2$, and thus, by Theorem~\ref{thm:kgreedy}, that a $k$-step OSP algorithm exists that returns a $p$-approximate solution for every $p$-system. Unfortunately, we will show that, if no constraint is given on the values in $D$, then no $k$-step OSP mechanism exists that is able to return a bounded approximation of the optimal solution, whenever $k < \left\lceil\frac{|D|}{2}\right\rceil - 2$.

\subsubsection{Two-Way Greedy Algorithm}
\label{subsec:matroids_algo}
We next describe in Algorithm~\ref{algo_greedy} the $k$-limitable two-way greedy algorithm for $p$-systems. For  readability, the algorithm assumes that $E = \{1, \ldots, n\}$ and $D_j = D = \{t_1, \ldots, t_d\}$ for every $j \in E$: if the domain of some agents were a subset of $D$, then clearly some of the queries could be skipped, allowing the actual implementation to be $k'$-limitable for them. 

The algorithms makes two kind of queries to agents, that we name respectively \emph{bottom queries} and \emph{top queries}. The former asks an agent if her type is the minimum in her current domain, and if so, it excludes the corresponding element from the current solution (see Algorithm~\ref{bottom}). Incidentally, excluding one element from the solution implies that there are other elements that cannot be excluded (e.g., in the minimum spanning tree application, once $\delta-1$ edges have been removed from a node of degree $\delta$, we know that the remaining edge will surely belong to the solution). Hence, when we run a bottom query, we do not only save those elements that have been excluded, but also the ones that consequently will surely belong to the solution (regardless of her type/weight): indeed, in both cases, we do not need to make further queries to these elements.

\begin{algorithm}[htbp]
 \DontPrintSemicolon
 \small
 \SetKw{True}{True}
 \SetKw{False}{False}
 \SetKwFunction{U}{unremovable}
 \SetKwFunction{BQ}{BQuery}
 \SetKwProg{Fn}{def}{:}{}
 \Fn{\BQ{$j$, $S$, $X$}}{
	ask $j$ if her type is $\min D_j$\;
	\If{yes}{
	 add $j$ to $X$\;
	 $U = \U{S, X}$\;
	 $S = S \cup U$\;
	 \lIf{$X \cup S = E$}{\Return \True}
	}
	\Return \False
 }
 \caption{Bottom Query}
 \label{bottom}
\end{algorithm}

In order to check which element cannot be excluded we assume that a function \U exists that given in input a partial solution $S$, and a set of excluded elements $X$, returns the set $U$ of elements that must be inserted in the solution. In general, this function may maintain the set $\mathcal{O}$ of maximal feasible solutions, exclude from $\mathcal{O}$ all solutions $S$ involving elements in $X$, and return the set $U = \{e \in E \setminus (X \cup S) \colon e \in T \; \forall\, T \in \mathcal{O}\}$, i.e., the set of those elements that belong to all non-excluded maximal solutions. 
Note that however the set $\mathcal{O}$ may in principle contain an exponentially large number of solutions, and thus the implementation of the function \U as described above is not in general polynomial. However, it is often possible, by exploiting the structure the problem, to implement  \U without the need of keeping this exponentially large data structure, as in the case of, e.g., minimum spanning tree.

The function is assumed to signal through the return value whether the final solution has been found, i.e., when all elements are  either in the current solution or have been excluded.

A top query instead asks an agent if her type is the maximum in her current domain, and if so, includes the corresponding element in the solution (see Algorithm~\ref{top}). As for the bottom query, this may in turn cause some other elements to be excluded from the solution (since they never appear in a solution with the one element that we just added to our solution).

\begin{algorithm}[htbp]
 \DontPrintSemicolon
 \small
 \SetKw{True}{True}
 \SetKw{False}{False}
 \SetKwFunction{Rem}{removable}
 \SetKwFunction{TQ}{TQuery}
 \SetKwProg{Fn}{def}{:}{}
 \Fn{\TQ{$j$, $S$, $X$}}{
	ask $j$ if her type is $\max D_j$\;
	\If{yes}{
	 add $j$ to $S$\;
	 $U = \Rem{S, X}$\;
	 $X = X \cup U$\;
	 \lIf{$X \cup S = E$}{\Return \True}
	}
	\Return \False
 }
 \caption{Top Query}
 \label{top}
\end{algorithm}

As above, we assume that a function \Rem exists that finds these elements; a general implementation would return the set $U = \{e \in E \setminus (X \cup S) \colon e \notin T \; \forall \, T \in \mathcal{O}\}$, where $\mathcal{O}$ is defined above, but more efficient implementations can in principle be found by exploiting the specific problem of interest. Finally, the function is assumed to signal through the return value whether the final solution has been found (i.e., when all elements are either in the current solution or have been excluded).

We are now ready to present in Algorithm~\ref{algo_greedy} the two-way greedy algorithm for finding the feasible solution of approximate maximum weight in a $p$-system. The algorithm works as follows: it maintains a set $A$ of alive agents, that is, those agents that did not reveal their type (i.e., they never answered yes to some query) and are not unremovable (i.e., they do not belong to any remaining maximal feasible solution); in Lines~\ref{startwhile1}-\ref{endwhile1}, the algorithm looks for the first agent whose type is not the minimum in her domain, by querying agents in order of their id until the first agent that answer no is found (or only one maximal feasible solution is left, that must be returned)\footnote{Note that in the algorithm $A$ changes as soon as either $S$ or $X$ change, and hence as soon as one agent answers yes to a query, but it does not change when the queried agent gives a negative answer.}; once the first agent is found whose type is not the minimum, in Lines~\ref{startfor}-\ref{endfor}, the algorithm asks each remaining agent whether her type is the minimum or the second minimum in her domain; finally, in Lines~\ref{startfirst}-\ref{endfirst}, the algorithm asks to the first alive agent (for whom the minimum was already discarded) if her type is the second or the third minimum in her domain; note that the algorithm tries to keep always one alive agent that is one query ahead with respect to the others, and thus if the first agent happens to not be  alive as an effect of the last two queries, we need to find another agent that will play this role, and this is done in Lines~\ref{startwhile2}-\ref{endwhile2}; this process is repeated until the domain of the first alive agent has size $3$ and the domain of remaining agents has size $2$, if $|D|$ is even, or the domain of the first agent has size $2$ and the domain of remaining agents has size $1$, otherwise; after that, one last query is asked (to all agents if $|D|$ is even, and to the first alive agent only otherwise) in order to have that the domain of all agents is fully revealed, and hence the desired maximal feasible solution can be immediately computed and returned.
\begin{algorithm}[htbp]
	\DontPrintSemicolon
	\small
	\SetKw{or}{or}
	\SetKw{and}{and}
	\SetKw{els}{else}
	Set $b = 1 - |D| \mod 2$\;
	Let $S = \U(\emptyset, \emptyset)$ \tcc*{The partial solution}
	Let $X = \Rem(\emptyset, \emptyset)$ \tcc*{The excluded elements}
	Use $A$ as an alias for $E \setminus (S \cup X)$ \tcc*{The set of alive elements}
	Let $j = 0$ \tcc*{The next agent to query}
	\While{$j \neq \min A$ \label{startwhile1}}{
		\lIf{\BQ{$\min A$, $S$, $X$}}{\Return $S$ \label{single1}}
		$j = j+1$ \label{endwhile1}
	}
	
	\While{$|D_{\min A}| > 2 + b$ \or $|D_{\max A}| > 1 + b$}{
		\For{$j \in A \colon j > \min A, |D_j| > 2$ \label{startfor}}{
			\lIf{\BQ{$j$, $S$, $X$}}{\Return $S$}
			\If{$j \in A$}{
				\lIf{\BQ{$j$, $S$, $X$}}{\Return $S$ \tcc*[f]{We can compress last two queries} \label{endfor}}
			}
		}
		$j = \min A$\;
		\If{$|D_j| > 2 + b$ \label{startfirst}}{
			\lIf{\BQ{$j$, $S$, $X$}}{\Return $S$}
			\If{$j \in A$}{
				\lIf{\BQ{$j$, $S$, $X$}}{\Return $S$ \tcc*[f]{We can compress last two queries} \label{endfirst}}
			}
			\While{$j \neq \min A$ \label{startwhile2}}{
				\lIf{\BQ{$\min A$, $S$, $X$}}{\Return $S$ \label{single2}}
				$j = j+1$ \label{endwhile2}
			}
		}
	}
	\If{$|D|$ is even}{
		\For{$j \in A \colon j > \min A$}{
			\lIf{\TQ{$j$, $S$, $X$}}{\Return $S$ \tcc*[f]{Only-minimum effective bottom query}}
		}
		\lIf{\BQ{$\min A$, $S$, $X$}}{\Return $S$ \tcc*[f]{Only-minimum effective bottom query} \label{single3}}
	}\Else{
		\lIf{\BQ{$\min A$, $S$, $X$}}{\Return $S$ \tcc*[f]{Only-minimum effective bottom query}}
	}
	add $\min A$ to $S$\;
	$U = \U(S, X)$\;
	\Return $S \cup U$
	\caption{$k$-limitable two-way greedy algorithm for a $p$-system $(E, \mathcal{F})$}
	\label{algo_greedy}
\end{algorithm}

We next prove the following result.
\begin{theorem}
\label{thm:matroids}
 Algorithm~\ref{algo_greedy} is two-way greedy, and it is $k$-limitable for $k \geq \left\lceil\frac{|D|}{2}\right\rceil - 2$. Moreover, it returns a $p$-approximation of the feasible set of maximum weight for every $p$-system.
\end{theorem}
\begin{proof}
In order to prove that  Algorithm~\ref{algo_greedy} is two-way greedy, it is sufficient to show that we do not interleave top and bottom queries to some element unless the domain of that element is revealable, but this is clearly the case for Algorithm~\ref{algo_greedy}, since interleaving only occurs when the domain of the element has size two.

In order to prove that the algorithm is $k$-limitable, it is sufficient to observe that queries can be compressed so that each agent receives at most $k+2$ queries, with the the $(k+2)$-th query allowed only if the previous queries are all bottom query, and being either a strongly only-minimum effective revelation query or an only-minimum effective bottom query.

To this aim, let us first assume that $|D|$ is even. Observe that the element $\min A$ receives one single query (at Line~\ref{single1} or at Line~\ref{single2}), two consecutive queries (that can be compressed) until the domain has size $3$ and a final single query (at Line~\ref{single3}). Hence she receives $2+x$ compressed queries, where $x$ is such that $|D| - 2x - 1 = 3$, and thus $x = \frac{|D|}{2} - 2 \leq k$, as desired. Moreover, it is immediate to check that the last query is an only-minimum effective bottom query. As for remaining elements they receive two consecutive queries until the domain has size $2$, and one final top query. Hence, they receive $1+x$ compressed queries, where $x$ is such that $|D| - 2x = 2$, and thus $x = \frac{|D|}{2} - 1 \leq k+1$, as desired. Moreover, for the last query, since the domain has size $2$, it is trivially an only-minimum effective bottom query.

If $|D|$ is odd, then element $\min A$ receives one single query (at Line~\ref{single1} or at Line~\ref{single2}), two consecutive queries (that can be compressed) until the domain has size $2$, and a final single query. Hence she receives $2+x$ compressed queries, where $x$ is such that $|D| - 2x - 1 = 2$, and thus $x = \frac{|D|+1}{2} - 2 \leq k$, as desired. Moreover, since the domain has size $2$, the last query trivially is an only-minimum effective bottom query. As for remaining elements they receive two consecutive queries until the domain has size $1$. Hence, they receive $x$ compressed queries, where $x$ is such that $|D| - 2x = 1$, and thus $x = \frac{|D|+1}{2}  - 1 \leq k+1$.

Finally, we prove the approximation of the mechanism. To this aim, observe that the outcome returned by this mechanism is exactly the same returned by a reverse greedy algorithm, even if the order in which elements are removed changes based on their type: for even $|D|$ it essentially alternates the order $(\min A, \min A + 1, \ldots, \max A)$ used for odd types (the first type in the domain, the third one, and so on) with the order $(\min A + 1, \ldots, \max A, \min A)$ used for even types.
\end{proof}

\subsubsection{Inapproximability Results}
We next show that no bounded approximation can be achieved by a $k$-step OSP mechanism whenever $k < \left\lceil\frac{|D|}{2}\right\rceil - 2$, even for $p$-systems with bounded $p$. Interestingly, on the way to prove this result, we will also highlight how reverse greedy algorithms appear to be more powerful than greedy algorithms when they should be implemented by selfish agents with limited cognitive abilities: indeed, while in Section~\ref{subsec:matroids_algo}, we proved that it is possible to turn a reverse greedy algorithm in a $k$-limited two-way greedy algorithm for every $k \leq \left\lceil\frac{|D|}{2}\right\rceil - 2$, we will observe that it is impossible to turn a greedy algorithm in a $k$-limited two-way greedy algorithm for every $k < |D|-3$.

To this aim, we consider the problem of maximizing the social welfare in a single-item auction, that, as described above, is a special case of problems on weighted matroids, and thus a special $p$-system.
Assume that for each agent $e \in E$, we have $w(e) \in D = \{t_1, \ldots, t_{d}\}$, where $d=|D| \geq 4$, and $t_{j+1} > \rho t_j$ for every $j = 1, \ldots, |D|-1$, for some $\rho > 0$.
Suppose that there exist a social choice function $f$ that is implementable by a $k$-step OSP auction and returns a feasible solution whose total weight is better than an $1/\rho$ fraction of the total weight of the optimal solution for any instance.
By Theorem~\ref{thm:kgreedy} $f$ must be implementable by a $k$-limitable two-way greedy mechanism $\M$ with implementation tree $\T$. We next provide some properties about $\T$. Recall that we denote as bottom query, a query in which the element is asked whether her weight is the smallest in the current domain, and in case of a positive answer, the element will surely not belong to the returned solution (and hence the element is not further queried). Similarly, we denote with top query, a query in which the element is asked whether her weight is the largest in the current domain, and in case of a positive answer, the element will surely belong to the returned solution (and hence the element is not further queried). Recall also that we say that interleaving occurs in $\T$ if there is a path such that element $e$ first receives a bottom query and then a top query, or viceversa: in a two-way greedy algorithm this is allowed only if at node $u$ in which the direction of queries interleaves the domain of $e$ is revealable, meaning that for all values in the domain except at most the smallest one, if queries before $u$ were bottom queries, and the largest one, otherwise, the outcome is always the same regardless the actions of other players.
\begin{lemma}
\label{lem:tree_prop}
 Let $u$ be a node of $\T$ and let $e = i(u)$. Suppose that $D_e(u)$ is not revealable. Let then $v$ be the unique child of $u$ such that $D_e(v)$ is not a singleton.

 If the query at $u$ is a bottom query, then $D_e(v) = \{t_x, t_{x+1}, \ldots, t_d\}$ for some $x > 1$.  Moreover, if $\min D_{e'}(v) < t_{x}$ for every $e' \neq e$, and there is at least $e''$ such that $\max D_{e''}(v) > t_{x+1}$, then $i(v) \neq e$.

 Similarly, if the query at $u$ is a top query, then $D_e(v) = \{t_1, \ldots, t_{x-1}, t_x\}$ for some $x < d$.  Moreover, if there is at least one agent $e''$ such that $\max D_{e''}(v) > t_{x}$, and $\min D_{e'}(v) < t_{x-1}$ for every $e' \neq e$, then $i(v) \neq e$.
\end{lemma}
\begin{proof}
 Suppose first that the query at $u$ is a bottom query.
 Since $D_e(u)$ is not revealable, and $\M$ is a two way greedy algorithm, this means that there has been no interleaving along the path from the root of $\T$ and $u$. Hence, all the query to $e$ preceding the one at $u$ must be bottom query, from which we achieve that $D_e(v) = \{t_x, t_{x+1}, \ldots, t_d\}$ for some $x > 1$, as desired.

 Now, suppose that the domain of the other agents is as in the claim. We now prove that the domain of $e$ at $v$ is not revealable, i.e., that for $w(e) = t_{x+1}$ there is both a leaf $\ell_{\textrm{in}}$ in the subtree rooted in $v$ such that $e$ belongs to the solution corresponding to $\ell_{\text{in}}$, and a leaf $\ell_{\textrm{out}}$ such that $e$ does not belong to the solution corresponding to $\ell_{\textrm{out}}$.
 Indeed, the profile $\y$ such that $y_e = t_{x+1}$ and $y_{e'} \leq t_{x}$ for every $e' \neq e$ is available at $v$ (and thus there must be a leaf $\ell_{\text{in}}$ at which it corresponds), and for it the mechanism must assign the item to $e$, otherwise the cost of the returned solution would be at most $t_{x} < \frac{1}{\rho} t_{x+1} \leq \frac{1}{\rho} OPT(\y)$, where $OPT(\y)$ denotes the optimal solution for profile $\y$. Similarly, the profile $\z$ such that $z_e = t_{x+1}$, $z_{e''} > t_{x+1}$, and $z_{e'} \in D_{e'}(v)$ for every $e' \neq e, e''$ is available at $v$, and for it the mechanism must not assign the item to $e$, otherwise the cost of the returned solution would be at most $t_{x+1} <\frac{1}{\rho} t_{x+2} \leq \frac{1}{\rho} OPT(\z)$.

 Since the domain of $e$ at $v$ is not revealable, it follows that we cannot interleave, and thus make a top query to $e$ at $v$. We next show that it is not possible to have a bottom query at this node, from which the claim follows. Suppose instead that the query at $v$ is a bottom query to $e$, and consider the profile $\y$ such that $y_e = t_x$ and $y_e' < t_x$ for every $e' \neq e$. Note that this profile is compatible with the children of $v$ corresponding to an yes answer to the bottom query. Hence, the mechanism does not assign the item to $e$, and thus receives a social welfare of at most $t_{x-1}< \frac{1}{\rho} t_x = \frac{1}{\rho} OPT(\y)$, contradicting the fact that the mechanism always returns a $\frac{1}{\rho}$-approximation.

 The case that the query at $u$ is a top query is very similar, and it is included only for sake of completeness.
 Since $D_e(u)$ is not revealable, and $\M$ is a two way greedy algorithm, this means that there has been no interleaving along the path from the root of $\T$ and $u$. Hence, all the query to $e$ preceding the one at $u$ must be top query, from which we achieve that $D_e(v) = \{t_1, \ldots, t_{x-1}, t_x\}$ for some $x < d$, as desired.

 Now, suppose that the domain of the other agents is as in the claim. We now prove that the domain of $e$ at $v$ is not revealable, i.e., that for $w(e) = t_{x-1}$ there is both a leaf $\ell_{\textrm{in}}$ in the subtree rooted in $v$ such that $e$ belongs to the solution corresponding to $\ell_{\text{in}}$, and a leaf $\ell_{\textrm{out}}$ such that $e$ does not belong to the solution corresponding to $\ell_{\textrm{out}}$.
 Indeed, the profile $\y$ such that $y_e = t_{x-1}$ and $y_{e'} < t_{x-1}$ for every $e' \neq e$ is available at $v$ (and thus there must be a leaf $\ell_{\text{in}}$ at which it corresponds), and for it the mechanism must assign the item to $e$, otherwise the cost of the returned solution would be at most $t_{x-2} < \frac{1}{\rho} t_{x-1} \leq \frac{1}{\rho} OPT(\y)$. Similarly, the profile $\z$ such that $z_e = t_{x-1}$, $z_{e''} > t_{x}$, and $z_{e'} \in D_{e'}(v)$ for every $e' \neq e, e''$ is available at $v$, and for it the mechanism must not assign the item to $e$, otherwise the cost of the returned solution would be at most $t_{x-1} <\frac{1}{\rho} t_{x} \leq \frac{1}{\rho} OPT(\z)$.

 Since the domain of $e$ at $v$ is not revealable, it follows that we cannot interleave, and thus make a bottom query to $e$ at $v$. We next show that it is not possible to have a top query at this node, from which the claim follows. Suppose instead that the query at $v$ is a top query to $e$, and consider the profile $\y$ such that $y_e = t_x$, $y_{e''} > t_x$ and $y_e' \in D_{e'}(v)$ for every $e' \neq e$. Note that this profile is compatible with the children of $v$ corresponding to an yes answer to the top query. Hence, the mechanism must assign the item to $e$, and thus receives a social welfare of at most $t_{x} < \frac{1}{\rho} t_{x+1} = \frac{1}{\rho} OPT(\y)$, contradicting the fact that the mechanism always returns a $\frac{1}{\rho}$-approximation.
\end{proof}

Despite the apparent symmetry between top and bottom queries in Lemma~\ref{lem:tree_prop}, it is the case that they are different in terms of $k$-step OSPness of the mechanism, as suggested by comparing Theorem~\ref{thm:matroids} (and the observation that Algorithm~\ref{algo_greedy} essentially provides an implementation of the reverse greedy algorithm for $p$-systems) and the following corollary of Lemma~\ref{lem:tree_prop}.
\begin{corollary}
\label{cor:greedy_fail}
 For every $\rho > 0$, every $d=|D| \geq 4$, every $k < d-3$, and every $n \geq d$,\footnote{Actually, it is possible to prove that greedy algorithms are weaker than reverse greedy in terms of implementability through $k$-step OSP mechanisms, even for smaller values of $n \geq 4$. Indeed, in this case, we would be able to prove inapproximability by $k$-limited mechanisms for every $k$ smaller than a threshold of about $\frac{2d}{3}$, and thus still larger than the $\frac{d}{2}$ threshold that we have been able to prove for reverse greedy algorithms. This results follows by noticing, as done in the proof below, that, by Lemma~\ref{lem:tree_prop}, agents can receive two consecutive top queries only in extreme cases, and they can never receive more than two consecutive queries: the result then follows by noticing that it is impossible that three different agents can be found in the extreme case allowing consecutive queries for more than one third of the total possible cases.} there is a $p$-system $(E, \mathcal{F})$ with $|E| = n$, such that any $k$-limitable two-way greedy implementation of the greedy algorithm for this problem, returns a feasible solution of total weight not larger than a $\frac{1}{\rho}$ fraction of the total weight of the optimal feasible solution.
\end{corollary}
\begin{proof}
 Let us consider the single-item auction setting with $D$ as described above.
 Any implementation of the greedy algorithm for this problem works by defining an ordering $\pi_x$ over agents for every $x \in [d]$, and making top queries about type $t_x$ according to the ordering $\pi_x$, until either the first agent is found that produces a yes answers (and it is then assigned the item) or the domain of some agent is revealable (and thus we can make a strongly only-maximum effective revelation query for that agent). Suppose that there exists one such implementation that is $k$-limitable and always returns an $\frac{1}{\rho}$-approximation of the optimal feasible solution.

 Let us focus on the path of the implementation tree of this mechanism compatible with the type of all agents being $t_1$.

 We first show that the domain of an agent $e$ cannot become revealable until the domain of remaining agents contains other elements except at most $t_1$ and $t_2$. Indeed, when the type of $e$ is $t_2$ we can still be in a profile in which $e$, by the approximation guarantee of the mechanism, must be assigned the item (i.e., when all remaining elements have type $t_1$), and in a profile in which $e$, by the approximation guarantee of the mechanism, must not be assigned the item (i.e., when there is at least one alternative agent whose type is $t_3$).
 Hence, in order for the domain of agent $e$ to be revealable we need that each remaining agent received at least $d - 2$ top queries. Note also that when the domain of an agent has size $2$, this is trivially revealable.

 Moreover, it must be the case that for each agent $e$ whose domain is $\{t_1, t_2\}$, we need to make another query unless the domain of all remaining agents is $\{t_1\}$. Indeed, if this query is not done, then the outcome received by $e$ is the same when her type is $t_1$ and the type of $e'$ is $t_2$, and when her type is $t_2$ and the type of $e'$ is $t_1$, where $e'$ is one of the remaining agents whose domain contains at least $t_1$ and $t_2$.  But this clearly lead to an approximation worse than $\frac{1}{\rho}$.

 Hence, if we show that there are at least two agents different from $e$ such that no two consecutive queries can be done to these agents, then the needed $d-2$ top queries cannot be compressed, and for at least one of these two agents, an additional query is needed. This agent then receives $d-1$ queries. Since the mechanism is $k$-limitable, then $d-1 \leq k +2$, that contradicts our hypothesis.

 However, the existence of these two agents is guaranteed by Lemma~\ref{lem:tree_prop} and our choice of $n$. Indeed, by Lemma~\ref{lem:tree_prop}, each agent that receives a top query about type $t_x$ cannot receives a top query about type $t_{x-1}$ unless all remaining agents discarded type $t_x$ from their domain. Hence, only the last agent that receiving the top query about $t_d, \ldots, t_3$ can immediately receive another query. Since there are at most $d-2$ such agents, and $n-1 \geq d$, the claim follows.
\end{proof}

We are now ready to prove that Theorem~\ref{thm:matroids} is tight.

\begin{theorem}
	For every $\rho > 0$, every $d=|D| \geq 5$, and every $k < \frac{d}{2}-2$, there is a $p$-system $(E, \mathcal{F})$ with $|E| = n$, such that no $k$-limitable two-way greedy algorithm returns a feasible solution of total weight not larger than a $\frac{1}{\rho}$ fraction of the total weight of the optimal feasible solution for this problem.
\end{theorem}
\begin{proof}
Let us consider the single-item auction setting with $D$ as described above.
Suppose that a two-way greedy algorithm exists that always returns an $\frac{1}{\rho}$-approximation of the optimum. By Corollary~\ref{cor:greedy_fail}, the algorithm cannot be the implementation of a greedy algorithm. Moreover, as showed above, if all agents receive as first query a top query, then the domain of these agents becomes revealable only if the domain has size two for all agents (except at most one). That is, any two-way greedy algorithm that makes a top query as first query to each agent must be an implementation of a greedy algorithm.

Hence, either the $k$-limitable two-way greedy algorithms starts with a bottom query for all agents, or there is an agent $e$ for which the first query is a bottom query and another agent $e'$ for which the first query is a top query. By Lemma~\ref{lem:tree_prop}, it follows that we cannot make a bottom query to $e$ about type $t_3$, and we cannot make a top query to $e'$ about type $t_4$. Hence, for each action of agents different from $e$ and $e'$, $e$ must receive the same outcome for each type in $\{t_3, \ldots, t_d\}$ and $e'$ must receive the same outcome for each type in $\{t_1, \ldots, t_4\}$: it is immediate to chech that there is a choice of action of other players (e.g., they all have type $t_1$), for which this constraint leads to an approximation worse than $\frac{1}{\rho}$.

Thus, the $k$-limitable two-way greedy algorithm must start with a bottom query for all agents.
Let us focus on the path of the implementation tree of this mechanism compatible with the type of all agents being $t_d$.

 We first show that the domain of an agent $e$ cannot become revealable until the domain of remaining agents contains other elements except at most $t_{d-1}$ and $t_d$. Indeed, when the type of $e$ is $t_{d-1}$ we can still be in a profile in which $e$, by the approximation guarantee of the mechanism, must not be assigned the item (i.e., when all remaining elements have type $t_1$), and in a profile in which $e$, by the approximation guarantee of the mechanism, must not be assigned the item (i.e., when there is at least one alternative agent whose type is $t_{d-2}$).
 Hence, in order for the domain of agent $e$ to be revealable we need that each remaining agent received at least $d - 2$ bottom queries. Note also that when the domain of an agent has size $2$, this is trivially revealable.

 Moreover, it must be the case that for each agent $e$ whose domain is $\{t_{d-1}, t_d\}$, we need to make another query unless the domain of all remaining agents is $\{t_d\}$. Indeed, if this query is not done, then the outcome received by $e$ is the same when her type is $t_d$ and the type of $e'$ is $t_{d-1}$, and when her type is $t_{d-1}$ and the type of $e'$ is $t_d$, where $e'$ is one of the remaining agents whose domain contains at least $t_{d-1}$ and $t_{d}$. But this clearly lead to an approximation worse than $\frac{1}{\rho}$.

 Hence, by Lemma~\ref{lem:tree_prop}, for each agent different from $e$ we cannot make more than two consecutive queries. Then we cannot compress the needed $d-2$ bottom queries in less than $\frac{d}{2}-1$ compressed queries, and for at least one agent, an adjunctive query is needed. This agent then receives $\frac{d}{2}$ queries. Since the mechanism is $k$-limitable, then $\frac{d}{2} \leq k +2$, that contradicts our hypothesis.
\end{proof}

%% file: num_query.tex
\section{An SOSP Mechanism with More than Two Queries}
\label{apx:num_query}
Consider two agents with domain $D_1 = D_2 = D = \{1, 2, 3, 4\}$.
We would like to implement the social choice function $f = (f_1, f_2)$ defined as follows:
$$
 f_1(t_1, t_2) = \begin{cases}
              4-t_1, & \text{if } t_1 \geq t_2;\\
              4-t_2, & \text{otherwise}.
             \end{cases}
$$
and
$$
 f_2(t_1, t_2) = \begin{cases}
              4-t_2, & \text{if } t_2 > t_1;\\
              0,     & \text{if } t_2 \leq t_1 \text{ and } t_1 = 4;\\
              4-t_1-1, & \text{otherwise}.
             \end{cases}
$$
Let us consider a mechanism $\M$ implemented as follows:
\begin{itemize}
 \item first ask to agent $1$ if her type is $4$, and if yes, return outcome $(0, 0)$, then ask to agent $2$ if her type is $4$, and if yes, return outcome $(0, 0)$;
 \item for $i = 2, 3$, first ask to agent $1$ if her type is $5-i$, and if yes, return outcome $(i-1, i-2)$, then ask to agent $2$ if her type is $5-i$, and if yes, return outcome $(i-1, i-1)$;
 \item return outcome $(3, 2)$.
\end{itemize}
It is immediate to check that $\M$ implements the social function $f$ and each agent receives $3$ queries that cannot be compressed. We next show that this mechanism can be augmented with payment making it SOSP.
\begin{lemma}
 There are payments for which the mechanism $\M$ is SOSP.
\end{lemma}
\begin{proof}
 Suppose that agent $i$ receives payment $0$ whenever her outcome is $0$, payment $4$ for outcome $1$, payment $7$ for outcome $2$, and payment $9$ for outcome $3$.

 Consider first agent $1$. Suppose that her type is $1$. At the first query, if she deviates and answers yes she receives outcome $0$ and payment $0$, corresponding to utility $0$. If she is truthful then either receives outcome $0$ and payment $0$, corresponding to utility $0$, or outcome $1$ and payment $4$ corresponding to utility $3 > 0$, or outcome $2$ and payment $7$, corresponding to utility $5 > 0$, or outcome $3$ and payment $9$, corresponding to utility $6 > 0$. Hence, the player has no incentive to deviate. At the second query, deviating and answering yes, guarantees utility $3$, whereas being truthful will guarantee utility $3$ or $5$ or $6$, and thus the agent has no incentive to deviate. At the third query, deviating and answering yes, guarantees utility $5$, whereas being truthful will guarantee utility $5$ or $6$, and thus the agent has no incentive to deviate.

 Suppose now that agent $1$ has type $2$. At first query, deviating gives utility $0$, while being truthful guarantee utility either $0$, or $4-2$, or $7-4$, or $9-6$. At the second query, deviating gives utility $4-2$, while being truthful guarantee utility either $4-2$, or $7-4$, or $9-6$. At the third query, deviating by answering no, gives utility either $7-4$ or $9-6$, but being truthful guarantees the same utility $7-4 = 3$.

 Suppose now that agent $1$ has type $3$. At first query, deviating gives utility $0$, while being truthful guarantee utility either $0$, or $4-3$, or $7-6$, or $9-9$. At the second query, deviating by answering no, gives utility either $4-3 = 1$, or $7-6 = 1$ or $9-9 = 0$, but being truthful guarantees an at least as good utility $4-3 = 1$.

 Finally, if agent $1$ has type $4$. At the first query, deviating by answering no, gives utility either $0$ or $4-4 = 0$, or $7-8 = -1$ or $9-12 = -3$, but being truthful guarantees an at least as good utility of $0$.

 The case of agent $2$ is symmetric and omitted.
\end{proof}